\PassOptionsToPackage{unicode}{hyperref}
\PassOptionsToPackage{hyphens}{url}
\PassOptionsToPackage{dvipsnames,svgnames,x11names}{xcolor}
\documentclass[
  12pt]{article}

\usepackage{amsmath,amssymb}
\usepackage{iftex}
\ifPDFTeX
  \usepackage[T1]{fontenc}
  \usepackage[utf8]{inputenc}
  \usepackage{textcomp} 
\else 
  \usepackage{unicode-math}
  \defaultfontfeatures{Scale=MatchLowercase}
  \defaultfontfeatures[\rmfamily]{Ligatures=TeX,Scale=1}
\fi
\usepackage{lmodern}
\ifPDFTeX\else  
\fi
\IfFileExists{upquote.sty}{\usepackage{upquote}}{}
\IfFileExists{microtype.sty}{
  \usepackage[]{microtype}
  \UseMicrotypeSet[protrusion]{basicmath} 
}{}
\makeatletter
\@ifundefined{KOMAClassName}{
  \IfFileExists{parskip.sty}{%
    \usepackage{parskip}
  }{
    \setlength{\parindent}{0pt}
    \setlength{\parskip}{6pt plus 2pt minus 1pt}}
}{
  \KOMAoptions{parskip=half}}
\makeatother
\usepackage{xcolor}
\makeatletter
\ifx\paragraph\undefined\else
  \let\oldparagraph\paragraph
  \renewcommand{\paragraph}{
    \@ifstar
      \xxxParagraphStar
      \xxxParagraphNoStar
  }
  \newcommand{\xxxParagraphStar}[1]{\oldparagraph*{#1}\mbox{}}
  \newcommand{\xxxParagraphNoStar}[1]{\oldparagraph{#1}\mbox{}}
\fi
\ifx\subparagraph\undefined\else
  \let\oldsubparagraph\subparagraph
  \renewcommand{\subparagraph}{
    \@ifstar
      \xxxSubParagraphStar
      \xxxSubParagraphNoStar
  }
  \newcommand{\xxxSubParagraphStar}[1]{\oldsubparagraph*{#1}\mbox{}}
  \newcommand{\xxxSubParagraphNoStar}[1]{\oldsubparagraph{#1}\mbox{}}
\fi
\makeatother

\usepackage{longtable,booktabs,array}
\usepackage{calc} 
\usepackage{etoolbox}
\makeatletter
\patchcmd\longtable{\par}{\if@noskipsec\mbox{}\fi\par}{}{}
\makeatother
\IfFileExists{footnotehyper.sty}{\usepackage{footnotehyper}}{\usepackage{footnote}}
\makesavenoteenv{longtable}
\usepackage{graphicx}
\makeatletter
\def\maxwidth{\ifdim\Gin@nat@width>\linewidth\linewidth\else\Gin@nat@width\fi}
\def\maxheight{\ifdim\Gin@nat@height>\textheight\textheight\else\Gin@nat@height\fi}
\makeatother
\setkeys{Gin}{width=\maxwidth,height=\maxheight,keepaspectratio}
\makeatletter
\def\fps@figure{htbp}
\makeatother

\makeatletter
\@ifpackageloaded{caption}{}{\usepackage{caption}}
\AtBeginDocument{%
\ifdefined\contentsname
  \renewcommand*\contentsname{Table of contents}
\else
  \newcommand\contentsname{Table of contents}
\fi
\ifdefined\listfigurename
  \renewcommand*\listfigurename{List of Figures}
\else
  \newcommand\listfigurename{List of Figures}
\fi
\ifdefined\listtablename
  \renewcommand*\listtablename{List of Tables}
\else
  \newcommand\listtablename{List of Tables}
\fi
\ifdefined\figurename
  \renewcommand*\figurename{Figure}
\else
  \newcommand\figurename{Figure}
\fi
\ifdefined\tablename
  \renewcommand*\tablename{Table}
\else
  \newcommand\tablename{Table}
\fi
}
\@ifpackageloaded{float}{}{\usepackage{float}}
\floatstyle{ruled}
\@ifundefined{c@chapter}{\newfloat{codelisting}{h}{lop}}{\newfloat{codelisting}{h}{lop}[chapter]}
\floatname{codelisting}{Listing}

\makeatother
\makeatletter
\@ifpackageloaded{caption}{}{\usepackage{caption}}
\@ifpackageloaded{subcaption}{}{\usepackage{subcaption}}
\makeatother

\ifLuaTeX
  \usepackage{selnolig}  
\fi
\usepackage[]{natbib}
\usepackage{bookmark}

\usepackage{amsthm}
\usepackage[spaceRequire=false]{algpseudocodex}
\usepackage{algorithm}
\usepackage{setspace}
\usepackage{bbm}

\newcommand{\epsilonv}{\boldsymbol{\epsilon}}
\newcommand{\epsilonc}{\epsilon}

\newcommand{\diag}{\mathbf{di}}

\newcommand{\pindext}{I^*}
\newcommand{\pindex}{I}
\newcommand{\pichain}{\overline{\pindex}}

\newcommand{\qoi}{\theta}

\newcommand{\np}{\psi}
\newcommand{\npset}{\psi}
\newcommand{\npi}{\npset_{\pindex}}
\newcommand{\npii}[1]{\npset_{\pindex_{#1}}}
\newcommand{\npscalar}{\psi_i}
\newcommand{\npsetl}{\check{\npset}}
\newcommand{\nptotal}{\npset}

\newcommand{\node}{\nu}
\newcommand{\nodel}{\check{\nu}}
\newcommand{\nodepath}{\overline{\node}}

\newcommand{\params}{\pi}
\newcommand{\data}{y}
\newcommand{\datam}{\boldsymbol{y}}
\newcommand{\datac}{y}
\newcommand{\datapot}{\widetilde{y}}
\newcommand{\datapotm}{\boldsymbol{\datapot}}
\newcommand{\datapotset}{\widetilde{y}}
\newcommand{\pdim}{D}
\newcommand{\pidim}{P}
\newcommand{\dpdim}{M}
\newcommand{\pathset}{\mathcal{P}}

\newcommand{\toynp}{\phi}
\newcommand{\toynpset}{\phi}
\newcommand{\toyqoi}{\theta}

\newcommand{\lfload}{\boldsymbol{\lambda}}
\newcommand{\lfloadnewc}{\widetilde{\lambda}}

\newcommand{\lffac}{\boldsymbol{\phi}}
\newcommand{\lfloadc}{\lambda}
\newcommand{\lfscale}{\mathbf{S}}
\newcommand{\lfcorfac}{\mathbf{R}}
\newcommand{\lfscalec}{S}
\newcommand{\lfcorfacc}{R}
\newcommand{\lffacc}{\phi}
\newcommand{\lffacnewc}{\widetilde{\phi}}
\newcommand{\lffacnew}{\boldsymbol{\lffacnewc}}
\newcommand{\effect}{\boldsymbol{\overline{\delta}}}
\newcommand{\effectc}{\overline{\delta}}
\newcommand{\effects}{\delta}

\newcommand{\pmmean}{\boldsymbol{\mu}}
\newcommand{\pmmeanc}{\mu}
\newcommand{\biasp}{\beta}

\newcommand{\E}{\mathbb{E}}
\newcommand{\V}{\mathrm{Var}}
\newcommand{\ind}{\mathrel{\perp\!\!\!\!\!\:\perp}}
\newcommand{\minsep}[2]{S^*_M\left(#1;#2\right)}

\newcommand{\Vhat}[1]{\mathrm{var}\left( #1 \right)}
\newcommand{\CVhat}[3]{\mathrm{cvar}_{#1,#2}\left( #3 \right)}
\newcommand{\Ehat}[3]{\hat{E}_{#1}\left( #2, #3_{#1} \right)}

\newcommand{\RR}{\mathbb{R}}

\newcommand{\branchplain}[2]{\underset{#1 \supset #2}{\mathrm{Branch}}}
\newcommand{\branch}[3]{\branchplain{#1}{#2}(#3)}
\newcommand{\subdivide}[4]{\underset{#1 \supset #2 \to #3}{\mathrm{Divide}}(#4)}
\newcommand{\merge}[4]{\underset{#2\leftarrow #1\rightarrow #3}{\mathrm{Merge}}(T)}

\newcommand{\STrees}{{Variance deltas} }

\newcommand{\TTrees}{{Variance Deltas} }
\newcommand{\Tree}{{variance delta} }
\newcommand{\Trees}{{variance deltas} }
\newcommand{\TreeNS}{{variance delta}}
\newcommand{\TreesNS}{{variance deltas}}

\newcommand{\DTree}{distributary tree }
\newcommand{\DTreeNS}{distributary tree}
\newcommand{\DTrees}{distributary trees }
\newcommand{\DTreesNS}{distributary trees}

\newcommand{\DDTrees}{Distributary Trees }
\newcommand{\DDTreesNS}{Distributary Trees}
\newcommand{\DDTree}{Distributary Tree }

\theoremstyle{plain}
\newtheorem{lemma}{Lemma}

\theoremstyle{plain}

\theoremstyle{definition}
\newtheorem{definition}{Definition}

\IfFileExists{xurl.sty}{\usepackage{xurl}}{} 
\hypersetup{
  pdftitle={Title},
  pdfauthor={Author 1; Author 2},
  pdfkeywords={3 to 6 keywords, that do not appear in the title},
  colorlinks=true,
  linkcolor={blue},
  filecolor={Maroon},
  citecolor={Blue},
  urlcolor={Blue},
  pdfcreator={LaTeX via pandoc}}

\newcommand{\anon}{1}

\begin{document}

\def\spacingset#1{\renewcommand{\baselinestretch}%
{#1}\small\normalsize} \spacingset{1}


\if1\anon
{
  \title{\bf \TTrees for Visualizing and Explaining Posterior Uncertainty}
  \author{Collin Cademartori\\
    Department of Statistical Sciences, Wake Forest University\\\href{mailto:cademac@wfu.edu}{cademac@wfu.edu}}
  \maketitle
} \fi

\if0\anon
{
  \bigskip
  \bigskip
  \bigskip
  \begin{center}
    {\LARGE\bf \TTrees for Visualizing and Explaining Uncertainty}
\end{center}
  \medskip
} \fi

\bigskip
\begin{abstract}
In observational settings, where the data generating process and possibly the sample size are not controlled, it is typically impossible to guarantee a priori that quantities of interest will be estimated with sufficient precision. However, even when the data do not determine the quantities of interest, they may still allow determination of what is missing -- unobserved information which, if observed, would meaningfully reduce uncertainty. We propose an interactive visualization system, termed variance deltas, to enable the discovery of such missing information from a Bayesian posterior distribution. This system, which we provide as a software package, represents missing information as subsets of unobserved model quantities, organized into a tree based on how well each subset explains uncertainty about the quantity of interest. This system both automates the construction of candidate subsets from minimal user input and implements interactive operations for the division and combination of subsets, allowing the efficient discovery of interesting and useful explanations. We demonstrate this system by using it to discover nonobvious explanations of uncertainty for (1) a treatment effect parameter in a simulated causal inference problem and (2) a population proportion in a forecasting model of real polling data with many sources of bias.
\end{abstract}

\noindent%
\vfill

\newpage
\spacingset{1.8} 

\section{Introduction}\label{sec:intro}

When modeling observational data where the data generating process is both uncontrolled and complex, statistical uncertainty about quantities of interest may be too great to permit any substantive inference \citep{greenland_identification,greenland_bias_models}. Common causes of this symptom include having insufficient amounts of data, having data of the wrong type, and having a nonidentified nuisance parameter on which a quantity of interest depends. Remedies may include collecting more and/or different data, bringing prior information to bear, or rethinking the substantive goals of the analysis. 

Finding the correct explanation for such uncertainty is often necessary to properly mitigate or accommodate it. But discovering such explanations typically requires manual analysis of the model and substantial guesswork. Furthermore, existing summarization and visualization tools often do not contain the information that we need. In Bayesian inference, \citet{bayesian_workflow} point out that the usual histogram- and scatter plot-based posterior visualizations ``do not fully capture the multiple levels of variation and uncertainty in our inferences.'' 

We propose an interactive visualization system to structure and partially automate the search for such explanations in the Bayesian setting. The user of this system only needs to specify a small set of initial potential explanations, from which further explanations are automatically generated and organized into a tree structure. The branches of these trees may then be interactively divided, extended, and merged to produce new explanations. In these trees, we think of information as flowing from the root to the leaf explanations, with each successive explanation weaker than the last. Consequently, we term these visualizations \TreesNS, in analogy with the tree-like flow of a river delta. We implement this system in a software package, \texttt{variance-deltas}, which we use to produce all figures.

\subsection{Quantifying Explanations}
To visualize uncertainty explanations, we must first measure explanatory power. The literature on explainable and interpretable machine learning has proposed many measures for explaining model predictions \citep{explanations_hullman,xai_review}. The popular SHAP method, for instance, essentially performs a sensitivity analysis of predictions with respect to the predictive features \citep{shap_original,shapley_sensitivity}. We follow a similar approach for explaining uncertainty.

We begin with a posterior density $p(\nptotal \mid\data)$ for a vector of unknowns $\nptotal\in \RR^\pdim$ given observed data $\data\in\RR^N$. Let $\pindext=\{1,2,\ldots\pdim\}$. For any $\pindex=\{i_1,i_2,\ldots, i_k\}\subseteq \pindext$, define the vector
\[
  \npset_{\pindex} = \left( \nptotal_{i_1}, \nptotal_{i_2},\ldots, \nptotal_{i_k} \right)\in\RR^k.
\] 
The total index set $\pindext$ is partitioned into two disjoint subsets:
\[
\pindext = \pindex_{\params}\cup\pindex_{\datapot},\quad\pindex_{\params}\in\RR^{\pidim},\;\pindex_{\datapot}\in\RR^{\dpdim},
\] 
where $\pidim + \dpdim = \pdim$. The unknowns $\npset_{\pindex_{\params}}$ are the model parameters, and the unknowns $\npset_{\pindex_{\datapot}}$ are unobserved data which could hypothetically be observed. For notational simplicity, we will write $\params = \npset_{\pindex_{\params}}$ and $\datapot = \npset_{\pindex_{\datapot}}$. We also isolate a distinguished parameter index $i_{\qoi}\in \pindex_{\params}$ which corresponds to a scalar quantity of interest $\qoi$, i.e. $\qoi = \nptotal_{i_{\qoi}}$. (If there are multiple scalar quantities of interest, we assume that they are taken one at a time.)

In practice, we often begin with a posterior $p(\params\mid\data)$ over just the parameters. When $\dim\left( \datapotset \right)>0$, the probability model must then be extended over this additional data, as in Bayesian models of missing data (e.g. for imputation) or replicated data (e.g. for posterior predictive model checking). 
In some cases, this extension may require additional parameters as well, such as observation-level random effects, in which case these are also included in $\params$.

Potential explanations of our uncertainty about quantity of interest $\qoi$ are represented by vectors $\nptotal_{\pindex}$ with $\pindex\subseteq \pindext\setminus\{i_{\qoi}\}$.
The degree to which $\npi$ explains uncertainty about $\qoi$ -- the explanatory power of $\npi$ -- can be quantified by the expected posterior variance of $\qoi$ given $\npi$, as a fraction of the overall posterior variance:
\begin{equation}
  \label{eq:unc_index}
  U^2_{\pindex} = \frac{\E\left[ \V\left( \qoi\mid \data, \npi \right)\mid\data \right]}{\V\left( \qoi\mid\data \right)}.
\end{equation}
We refer to $U^2_{\pindex}$ as the uncertainty index of $\npi$. By the law of total variance, $0 \leq U^2_{\pindex}\leq 1$, and so $\npi$ is a powerful explanation if $U^2_{\pindex}\ll 1$. We can also write the uncertainty index as
\begin{equation}
  \label{eq:sobol_form}
  U^2_{\pindex} = 1 - \frac{\V\left[ \E\left( \qoi\mid \data, \npi \right)\mid\data \right]}{\V\left( \qoi\mid\data \right)}.
\end{equation}
The rightmost term is a form of the Sobol' index \citep{sobol_index}, which measures the sensitivity of the conditional mean $\E\left( \qoi\mid \data, \npi \right)$ to $\npi$. Thus, in our framework, locating powerful explanations $\npi$ involves analyzing the sensitivity of $\qoi$ to the components of $\nptotal$.

\subsection{Toy Example: Normal-Normal Model}\label{sec:toy_model}

We introduce \Trees with a simple hierarchical model.
Suppose that we obtain $N$ observations from each of $J$ groups, yielding data $\datam = \left(\datac_{ij}\right)_{i,j=1}^{N,J}$. Consider the following model:
\begin{gather}
    y_{1j},\ldots, y_{Nj} \stackrel{iid}{\sim} \mathrm{normal}(\toynp_j, \sigma_j)\text{ for } 1\leq j\leq J,\label{eq:toy_hierarchical_lik}\\
    \toynp_1,\ldots,\toynp_J \stackrel{iid}{\sim} \mathrm{normal}(\toyqoi, \tau),\quad
    \toyqoi \sim\mathrm{normal}(0, 1).\label{eq:toy_hierarchical_prior}
\end{gather}

We take the superpopulation mean $\toyqoi$ to be our quantity of interest. We treat $\tau$ and the $\sigma_j$ as fixed hyperparameters for simplicity. For this example, we simulate data $\data$ from \eqref{eq:toy_hierarchical_lik} with $N=2$ observations from each of $J=3$ groups. We then generate samples from the posterior $p(\toynp_1,\toynp_2,\toynp_3,\toyqoi\mid \datam)$. Table \ref{fig:toy_hyper} summarizes the choices of hyperparameters for these steps.

\begin{table}
  \centering
  \renewcommand{\arraystretch}{1.5}
  \begin{tabular}{|c|c|c|}
    \hline
    \multicolumn{3}{|c|}{$\tau=2^{-2}$}\\\hline
    $\sigma_1 = 2^{-1}$ & $\sigma_2 = 2^0$ & $\sigma_3 = 2^1$\\ \hline
    $\toynp_1 = -1/2\; (\ast)$ & $\toynp_2 = 0\; (\ast)$ & $\toynp_3 = 1/2\; (\ast)$\\ \hline
  \end{tabular}
  \caption{Hyperparameters for simulating data and posterior sampling from the model \eqref{eq:toy_hierarchical_lik} and \eqref{eq:toy_hierarchical_prior}. Quantities marked with $(\ast)$ are treated as random for posterior sampling.}
  \label{fig:toy_hyper}
\end{table}

The parameters $\params$ consist of the overall mean $\qoi$ and group-level means $\toynpset = \left(\toynp_1,\toynp_2,\toynp_3\right)$. The hypothetical data $\datapotm=\left(\datapot_{ij}\right)_{i,j=1}^{N,J}$ are then sampled from \eqref{eq:toy_hierarchical_lik} independently of $\datam$. Introducing these hypothetical data points allows us to understand the impact of the sample size $N$ on our uncertainty for $\toyqoi$. Overall, our unknowns are $\nptotal = \left[ \params, \datapotset \right]$, where $\left[ u,v \right]$ denotes the concatenation of the vectors $u$ and $v$.
 
A \Tree is a tree rooted at the quantity of interest $\qoi$ with each non-root node representing a possible explanation $\npi$ with $\pindex\subseteq \pindext\setminus\{i_{\qoi}\}$. 
Figure \ref{fig:toy_tree_ex} displays one possible \Tree for $\qoi$. We interpret this tree as follows:
\begin{itemize}
  \item Each node is labelled with the corresponding model quantities in $\npi$. If a quantity is multidimensional (e.g. a vector or matrix), then the indices of included components are indicated next to the quantity name. For instance, the group-level mean $\toynp_1$, which is the first component of the vector $\toynp$, is indicated with the label \includegraphics[height=0.02\textheight]{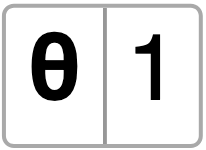}.
  \item The leaves correspond to the hypothetical data from each group $\datapot_j = \left( \datapot_{ij} \right)_{i=1}^N$. 
  \item The paths $\toyqoi\rightarrow \toynp_i\rightarrow \datapot_j$ reflect the hierarchical structure of the model. In particular, each path is a Markov chain in that $\toyqoi\ind\datapot_j\mid \toynp_j$.
  \item The horizontal position of a node containing explanation $\npi$ is given by the root uncertainty index $U_{\pindex}=\sqrt{U^2_{\pindex}}$ (where the square root is taken so that $U_{\pindex}$ is on the standard deviation scale). 
\end{itemize}

\begin{figure}
  \centering
  \includegraphics{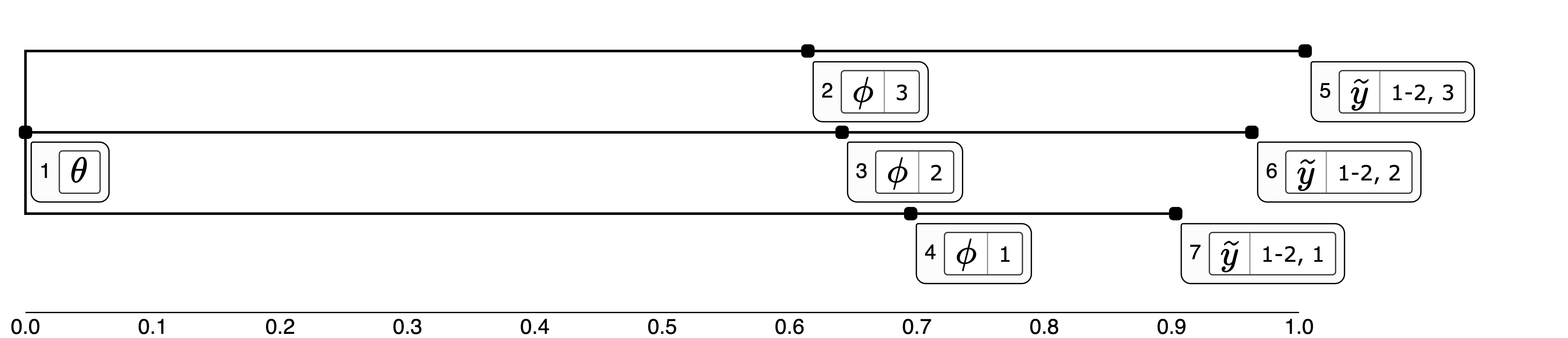}
  \caption{A \Tree constructed for the toy model given in \eqref{eq:toy_hierarchical_lik} and \eqref{eq:toy_hierarchical_prior}.}
  \label{fig:toy_tree_ex}
\end{figure}

This first tree shows that all three groups exert similar influence on our uncertainty about $\toyqoi$. 
In particular, exact knowledge of any one $\toynp_j$ cannot be expected to reduce posterior uncertainty about $\toyqoi$ to less than $\approx 60\%$ of its original level. This raises the question: how much more information would be obtained by learning \textit{all} group-level means at once. 

This question is answered by adding a node to our \Tree corresponding to the combined group-level means $\toynpset$. Figure \ref{fig:toy_tree_ex2} displays the result. Learning the entire vector $\toynpset$ reduces expected posterior uncertainty to about $40\%$ of its current level -- a substantial improvement. However, to reduce our uncertainty by more than half, this tree also suggests that we would need data from additional groups.

\begin{figure}
  \centering
  \includegraphics{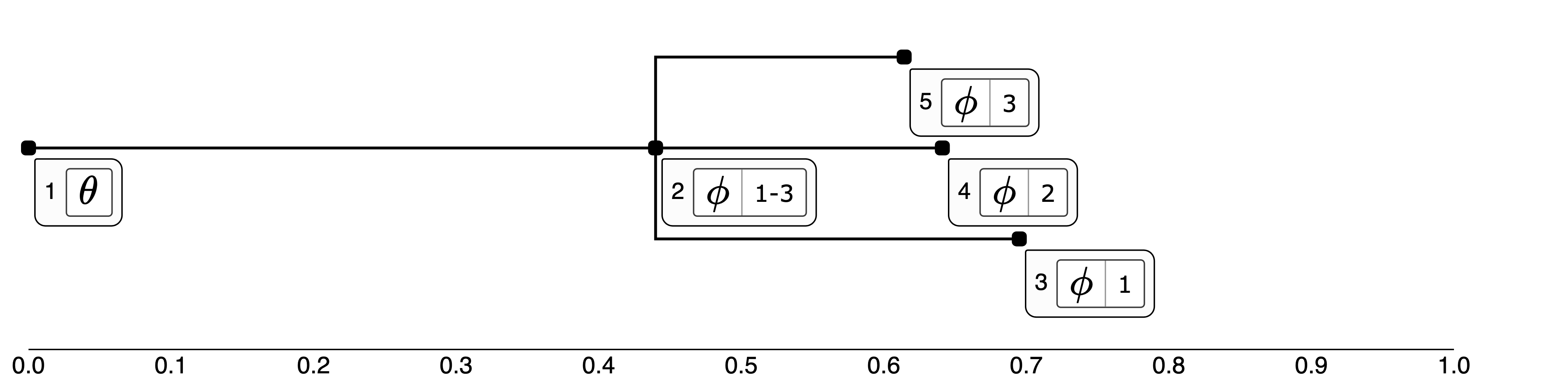}
  \caption{The tree from Figure \ref{fig:toy_tree_ex} with a new node added for the set $\toynpset = (\toynp_1,\toynp_2,\toynp_3)$.}.
  \label{fig:toy_tree_ex2}
\end{figure}

This example demonstrates how \Trees can both answer initial questions (``Which group is most important for our uncertainty about $\toyqoi$?'') and redirect attention to more important questions (``How much can we learn using data from only three groups?'').
In the more complex models of our two later examples, \Trees are used to discover explanations that are both more useful and harder to anticipate.

\subsection{Related Work}

This work follows a growing literature on Bayesian workflow, which addresses all stages of the statistical modeling pipeline in an integrated manner. 
Other applications of static or interactive graphics to problems in Bayesian workflow include the shinystan and bayesplot packages \citep{shinystan, bayesplot} -- tools for plotting MCMC outputs and diagnostics. Outside the Bayesian setting, graph-based visualizations have been used to reveal model structure and variable importance \citep{mv_bayes_nets,mv_variable_import}.

Previous work has also applied tools from sensitivity analysis to Bayesian workflow problems. Much of this work studies the sensitivity of inference to non-random aspects of the analysis, such as the prior specification \citep{sensitivity_prior1,sensitivity_prior2} and the data sample \citep{data_sensitivity}. Decision-theoretic applications, on the other hand, analyze sensitivity with respect to random, modeled quantities, as we do here. For instance, the expected value of partial perfect information is defined as the reduction in expected loss when a decision is made with and without conditioning on some unknown quantity \citep{evppi}. When the loss is quadratic, this is closely related to our uncertainty index \eqref{eq:unc_index} \citep{evppi_app}. Whereas such value of information analyses generally assume that we know which quantities we want to condition on in advance, this work is concerned with discovering interesting sets of unknowns to analyze.

The structure of our proposed visualizations is also inspired by John Tukey's early work on interactive statistical graphics \citep{tukey_cognostics}. In moving beyond direct visualization of posterior samples, we are influenced by Tukey's suggestion that, to invent useful new forms of visualization, we should ``find new phenomena of potential interest, and then learn how to make displays that will be likely to reveal them'' \citep{tukey_decades}. 

\subsection{Notation}
Whenever $S$ is a set, $\mathcal{P}(S)$ denotes the power set of $S$. We abuse notation slightly throughout this paper, using the same letter (e.g. $\npset$) for both a random variable (as in $\mathrm{Var}(\npset)$) and a possible value of that random variable (as in the density expression $p(\npset)$). Ambiguities are resolved by context or otherwise explicitly disambiguated.

We assume that model parameters $\npset$ and data $\data$ are vectors. Occasionally, we treat some components of the parameters and/or data as matrices, which we write with boldface characters. When we think of such a matrix as part of the vectors $\nptotal$ or $\data$, we assume the matrix is ``flattened'' into a vector (e.g. in column-major order, though this will not be specified). The $i^{\mathrm{th}}$ row or column of a matrix $\mathbf{A}$ is denoted $A_{i\cdot}$ and $A_{\cdot i}$, respectively.

\subsection{Outline}

The remainder of this paper is organized as follows. Section \ref{sec:post_trees} defines \TreesNS, presents an algorithm for generating them semi-automatically, and introduces a set of interactive operations for transforming them. The next two sections apply \Trees to find nontrivial explanations of uncertainty. Section \ref{sec:ex1_sc} applies our method to a simulated causal inference problem using panel data, where the quantity of interest is a treatment effect. Then, in Section \ref{sec:ex2_polls}, we examine a polling model applied to real data from the 2016 U.S. presidential election, where the quantity of interest is the Democratic share of the electorate on election day. Section \ref{sec:conclusion} presents concluding thoughts.

\section{\TTrees}\label{sec:post_trees}

\STrees are constructed from two components: the underlying tree which encodes conditional independence relationships in the model and a mapping from tree nodes to estimates of the corresponding uncertainty indices. The formal presentation of \Trees is divded into four parts:
(i) the definition of the \textit{\DTreesNS} underlying them, (ii) a method for estimating the uncertainty indices needed to plot them, (iii) an algorithm for semi-automatically constructing them, and (iv) the interactive operations that transform and extend them.

\subsection{\DDTrees}

A \Tree is a visualization of a tree rooted at the quantity of interest $\qoi$ where the non-root nodes represent possible explanations $\npi$ with $\pindex\subseteq\pindext\setminus\{i_{\qoi}\}$. Maintaining our hydrological analogy, we term these \DTreesNS \footnote{In a river delta, a distributary is an outflowing branch of the river which does not rejoin the main stream of the river.}. The paths in a \DTree are Markov chains (in a sense which we define below), and thus represent a diminishing flow of information from the start to the end node.

\begin{definition}[\DDTreesNS]
  \label{def:sensitivity_tree}
  Let $T = \left( \mathcal{N}, P\right)$ be a tree, where
  \begin{itemize}
    \item $\mathcal{N}$ is the set of tree nodes. Each node $\node\in\mathcal{N}$ is associated with a subset $\pindex(\node)\subseteq\pindext$. To simplify notation, we write $\npset(\node)$ for $\npset_{\pindex(\node)}$, the unknowns associated to $\node$.
    \item $P:\mathcal{N}\setminus\{\node_0\}\to\mathcal{N}$ maps each non-root node to its unique parent, where $\node_0$ is the root. The undirected graph with vertices $\mathcal{N}$ and edges $\{\{P({\node}), \node\}\}_{\node\in\mathcal{N}\setminus\{\node_0\}}$ must be acyclic and connected, so that $P$ yields a well-defined tree.
  \end{itemize}
  Then $T$ is a \DTree for quantity of interest $\qoi$ if it satisfies the following constraints.
  \begin{enumerate}
    \item (\textbf{Rooted at $\boldsymbol{\qoi}$}) We must have $\pindex(\node_0) = \{i_{\qoi}\}$.
    \item (\textbf{Markov Paths}) A path is a sequence $(\node_1,\ldots,\node_k)$ with $\node_i = P(\node_{i+1})$ for $1\leq i\leq k-1$. Every path with $k\geq 2$ must be a Markov chain in the sense that
    \begin{equation}
      \label{eq:markov_path_condition}
      \npset(\node_{\leq i}) \ind \npset(\node_{>i}) \mid \npset(\node_i) \quad\text{for all}\quad 1\leq i\leq k-1,
    \end{equation}
    where $\npset(\node_{\leq i}) = \npset_{I_{\leq i}}$ with $I_{\leq i} = \cup_{j=1}^i I(\node_j)$ and $\npset_{j>i}=\npset_{I_{> i}}$ with $I_{> i} = \cup_{j={i+1}}^k I(\node_j)$.
    \item (\textbf{Unique Children}) For any $\node, \node'\in\mathcal{N}$, if $P(\node) = P(\node')$, then $\pindex(\node) \neq \pindex(\node')$.
  \end{enumerate}
\end{definition}

As we will see momentarily, constraints 1 and 2 above are the basic properties that make \DTrees useful. Constraint 3 implies that distinct rooted paths in $T$ correspond to distinct Markov chains. Note also that, subject to constraint 3, $\pindex(\cdot):\mathcal{N}\to \mathcal{P}(\pindext)$ need not be injective, so the same model quantities $\npi$ can occur at different nodes of the same distributary tree.

\subsection{Visualizing and Estimating Uncertainty}

When we plot \DTrees in a \TreeNS, the horizontal position of each node $\node$ is determined by an estimate of its root uncertainty index $U_{\pindex(\node)}$, as in Figure \ref{fig:toy_tree_ex}. It follows from Definition \ref{def:sensitivity_tree} that $U_{\pindex(\node)}$ decreases from child to parent nodes.

\begin{lemma}
   Let $T$ be a \DTree. Then $U^2_{\pindex(\node)}$ decreases from child to parent nodes in $T$, and $U^2_{\pindex(\node_0)}=0$, where $\node_0$ is the root of $T$. 
\end{lemma}
\begin{proof}
  It follows directly from the definition \eqref{eq:unc_index} that $U^2_{\pindex(\node_0)}=U^2_{\{i_\qoi\}} = 0$. Next, suppose $\node = P(\node')$. Then by passing successively from child to parent nodes, we can construct a path beginning at $\node_0$ and terminating at $\node'$ which contains $\node$. Constraint 2 then implies that $\qoi\to \npset(\node)\to \npset(\node')$ is a Markov chain, and the law of total variance gives
  \begin{align*}
    \E\V\left( \qoi\mid \npset(\node') \right) &= \E\V\left( \qoi\mid \npset(\node), \npset(\node') \right) + \E\V\left( \E\left[ \qoi\mid \npset(\node), \npset(\node') \right] \mid \npset(\node')\right)\\
    &\geq \E\V\left( \qoi\mid \npset(\node), \npset(\node') \right)\\
    & = \E\V\left( \qoi\mid \npset(\node)\right),
  \end{align*}
  where the final equality follows by the conditional independence of $\npset(\node')$ and $\qoi$ given $\npset(\node)$. Plugging the first and last terms above into \eqref{eq:unc_index} immediately yields $U^2_{\pindex(\node')}\geq U^2_{\pindex(\node)}$.
\end{proof}

When working with approximate posterior samples $\mathcal{S} = \{\np^{(s)}\}_{s=1}^S$ (as produced by any MCMC algorithm), we estimate the denominator of \eqref{eq:unc_index} by the sample variance $\mathrm{var}\left( \mathcal{S} \right) = \frac{1}{S-1}\sum_{s=1}^S(\qoi^{(s)} - \bar{\qoi})^2$. To estimate the numerator of \eqref{eq:unc_index}, we use a data splitting strategy. Assuming $S$ is even, we use the first $S/2$ samples $\mathcal{S}_E$ to construct an estimator $\hat{E}_{\pindex}\left( \mathcal{S}_E,\cdot \right)$
of the conditional expectation $\E\left[ \qoi\mid\npi=\cdot, \data \right]$. Then, using the last $S/2$ samples $\mathcal{S}_V$, we estimate the expected conditional variance $\E\left[ \mathrm{Var}\left( \qoi\mid\data, \npi \right) \right]$ as
\begin{equation}
  \label{eq:unc_num_est}
  \CVhat{\pindex}{\hat{E}}{\mathcal{S}}= \frac{2}{S}\sum_{\npset^{(s)}\in \mathcal{S}_V}\left( \qoi^{(s)} - \Ehat{\pindex}{\mathcal{S}_E}{\npset^{(s)}} \right)^2.
\end{equation} 

With these, our estimate of the uncertainty index \eqref{eq:unc_index} is then just
\begin{equation}
  \label{eq:unc_est}
  \hat{U}^2_{\pindex} = \frac{\CVhat{\pindex}{\hat{E}}{\mathcal{S}}}{\Vhat{\mathcal{S}}}.
\end{equation}
Splitting the data in this way eliminates overfitting bias in the conditional variance estimate \eqref{eq:unc_num_est}. In particular, if the expected mean squared error (MSE) of $\hat{E}$ vanishes asymptotically, then \eqref{eq:unc_num_est} is asymptotically unbiased for the conditional variance (proof in supplementary materials). In most situations, an $\hat{E}$ with vanishing MSE will not be available, and \eqref{eq:unc_num_est} will be biased upward by the expected MSE. However, we find in our experiments that, as long as $S$ is large enough that the variance of $\hat{E}$ small, the performance of \eqref{eq:unc_est} is sufficient.

The construction \eqref{eq:unc_est} still requires us to choose an estimator $\hat{E}$. Throughout our examples, we take 
$\hat{E}_{\pindex}\left( \mathcal{S}_E,\cdot \right)$ to be a random forest regressor, chosen primarily for its ability to handle nonlinearities and interactions without explicit modeling. However, we emphasize that \eqref{eq:unc_est} can be used with any choice of $\hat{E}$.

\subsection{Constructing \DDTreesNS}

Constructing useful \DTrees faces two difficulties, which we address in this section.
\begin{enumerate}
  \item Constructing \DTrees manually requires the user to know both (i) which unknowns $\npi$ should be selected for the nodes of the tree and (ii) how to connect the nodes so that the Markov condition \eqref{eq:markov_path_condition} is satisfied for every path in the tree. We resolve this problem by defining an algorithm to construct \DTrees from a specification of just the root node (the quantity of interest) and the leaf nodes.
  \item For any given model, the root and leaf nodes alone usually do not determine a unique valid \DTreeNS. This is an instance of a more general problem that arises in sensitivity analysis: the need to constrain the subsets of quantities with respect to which the analysis will be performed \citep{sensitivity_subspaces,sobol_group_inputs,global_sens_rev}. Therefore, we need an additional criterion that allows us to choose between multiple options for constructing the tree at any given stage of the algorithm. For this, we introduce a measure of explanatory specificity.
\end{enumerate}

To resolve these difficulties, we need the concepts of factorizations and factor graphs.

\begin{definition}[Factorization]
  Let $h(\nptotal)$ be a nonnegative function. Then $h$ factorizes over $\mathcal{F}\subseteq\mathcal{P}(\pindext)$ (equivalently, $\mathcal{F}$ is a factorization of $h$), if, for functions $f_{\pindex}:\RR^{|\pindex|}\to\RR_+$,
  \begin{equation}
    \label{eq:post_fac}
    h(\nptotal) = \prod_{\pindex\in\mathcal{F}} f_{\pindex}(\npset_{\pindex}).
  \end{equation}

  Suppose $h(\nptotal)=p(\nptotal\mid\data)$, the posterior density for a joint model $p(\nptotal, \data)$, and that $p$ factorizes over $\mathcal{F}$. Then $\mathcal{F}$ is a Bayesian factorization if $\mathcal{F} = \mathcal{F}_L\cup\mathcal{F}_P$ with $\mathcal{F}_L\cap\mathcal{F}_P=\emptyset$, and where:
  \begin{enumerate}
    \item The prior $p(\params)$ factorizes over $\mathcal{F}_P$.
    \item The likelihood (for observed and hypothetical data) $p(\data, \datapot\mid \params)$ factorizes over $\mathcal{F}_L$.
  \end{enumerate}
\end{definition}

A factorization $\mathcal{F}$ can be conveniently represented using an undirected graph.
\begin{definition}[Undirected Graph]
  An undirected graph $M$ is a pair $(V_M, E_M)$, where $V_M$ is a set of vertices and $E_M$ is a set of undirected edges. Specifically:
  \[
    E_M \subseteq \left\lbrace \{x,y\} \;\vert\; x,y\in V_M, x\neq y \right\rbrace.
  \]
  An undirected graph is bipartite if the vertices can be partitioned as $V_M=V_1\cup V_2$, where $V_1\cap V_1 = \emptyset$, and
  \[
    E_M \subseteq \left\lbrace \{x,y\} \;\vert\; x\in V_1,y\in V_2 \right\rbrace,
  \]
  i.e. all edges connect a vertex in $V_1$ with a vertex in $V_2$.
\end{definition}

\begin{definition}[Factor Graph]
  The factor graph $G_{\mathcal{F}}$ corresponding to a factorization $\mathcal{F}$ is defined as
  \[
    V_{G_{\mathcal{F}}} = \pindext \cup \mathcal{F},\quad E_{G_{\mathcal{F}}} = \left\{ \left\{ i, \pindex \right\} \;\Big\vert\; \pindex\in\mathcal{F}, i\in\pindex \right\}.
  \]
  In words, $G_{\mathcal{F}}$ is a bipartite graph with a vertex for every factor $\pindex\in\mathcal{F}$ and index $i\in\pindext$, with an edge between $i$ and $\pindex$ if and only if $i\in\pindex$.
\end{definition}

\subsubsection{Measuring the Specificity of an Explanation}

If we want to minimize $U^2_{\pindex}$, the best we can do is to observe all unknown quantities other than $\qoi$. In other words, we have
\begin{equation}
  \label{eq:uncertainty_index_min}
  \pindext\setminus\{i_{\qoi}\} \in \underset{\pindex\subseteq \pindext\setminus\{i_{\qoi}\}}{\mathrm{arg}\min}\; U^2_{\pindex}.
\end{equation}
Despite its maximal explanatory power, $\npset_{\pindext\setminus\{\qoi\}}$ cannot tell us which quantities \textit{specifically} explain our uncertainty about $\qoi$ and is thus not a useful explanation. To generate trees containing potentially useful explanations, we need to exclude explanations that are overly non-specific. This also provides a useful criterion to resolve the second difficulty above: When faced with multiple choices for the next node in a tree, we can always choose the node $\node$ with more specific explanation $\npset(\node)$. We quantify specificity in terms of the distance between the unknowns $\npset(\node)$ and likelihood factors in the posterior factor graph.

Let $\mathcal{F}$ be a Bayesian factorization of the posterior density $p(\nptotal\mid \data)$ with $G_{\mathcal{F}}$ the corresponding factor graph. 
We start with a few desiderata:
\begin{enumerate}
  \item If $\pindex'\subseteq \pindex$, then $\npset_{\pindex'}$ must be at least as specific as $\npset_{\pindex}$.
  \item An explanation $\npi$ is maximally specific if the quantities in $\npi$ are involved in exactly one likelihood factor, i.e. if there is exactly one $\pindex'\in\mathcal{F}_L$ with $\pindex\subseteq \pindex'$. These explanations point to a particular quantity of (hypothetically) observable information which is insufficient or missing.
  \item Conversely, $\npi$ is minimally specific if $\npi$ is equally involved in modeling every data point, as with a global mean or scale parameter. Such explanations do not suggest particular remedies, since \textit{any} additional data could be relevant to learning $\npi$.
\end{enumerate}

Our specificity measure interpolates these extremes by measuring the distance between explanations $\npi$ and likelihood factors in $G_\mathcal{F}$.
A path in $G_{\mathcal{F}}$ is any sequence of (distinct) vertices $\overline{v} = (v_1,v_2,\ldots,v_k)$ where $(v_j,v_{j+1})$ are joined by an edge for all $1\leq j\leq k-1$. Let $\pathset_k(G_{\mathcal{F}})$ be the set of all paths in $G_{\mathcal{F}}$ of length $k$. For each $i\in \pindext$, the depth $d_{\mathcal{F}}(i)$ of $i$, is defined as the length of the shortest path from $i$ to a likelihood factor. Symbolically:
\begin{equation}
  \label{eq:param_depth}
  d(i) = \min_{f\in\mathcal{F}_L} d(i, f), \text{ where } d(i,f) = \min\left\lbrace k\;\Big\vert\; \exists \overline{v}\in \pathset_k(G_{\mathcal{F}})\;s.t.\; \overline{v}_1=i, \overline{v}_{k}=f \right\rbrace.
\end{equation}

Similarly, the set of most relevant likelihood factors for $i$ is defined as
\begin{equation}
  \label{eq:most_rel}
  \mathcal{F}_L(i) = \left\lbrace f\in \mathcal{F}_L\mid d(i,f) = d(i)  \right\rbrace.
\end{equation}
Finally, for any index set $\pindex\subseteq\pindext$, the likelihood specificity is defined as
\begin{equation}
  \label{eq:lik_complexity}
  S_{\mathrm{lik}}(\pindex) = \frac{1}{ |\mathcal{F}_L| } \left|\bigcup_{i\in\pindex }\mathcal{F}_L(i)\right|,
\end{equation}

In words, sets $\pindex$ with low $S_{\mathrm{lik}}(\pindex)$ index unknowns $\npi$ which are most relevant to a relatively small set of likelihood factors. Furthermore, it is easy to see that likelihood specificity satisfies the above three desiderata. Note in particular that $S_{\mathrm{lik}}(\pindex)$ is minimized when $\npi$ occurs in exactly one likelihood factor, and is maximized when some $\npset_i$ with $i\in\pindex$ is a global quantity, in the sense that $d(f,i)=d(i)$ for all $f\in\mathcal{F}_L$. 

\subsubsection{Generating \DDTrees From Leaves}

Given our measure of specificity $S_{\mathrm{lik}}$, we may now define a procedure for generating \DTreesNS. To do this, we need to leverage information about the dependencies between model quantities, which can be derived from the factor graph $G_{\mathcal{F}}$. Specifically, for a factor graph $G_{\mathcal{F}}$, define the Markov graph $M_{\mathcal{F}}$ with vertices $V_{M_{\mathcal{F}}} = \pindext$ and edges
\[
  E_{M_{\mathcal{F}}} = \left\lbrace \{i_1,i_2\} \;\Big\vert\; \exists \pindex\in\mathcal{F} \text{ s.t. } i_1,i_2\in \pindex \right\rbrace.
\]
In other words, we connect two indices in $M_{\mathcal{F}}$ if they belong to a common factor in $\mathcal{F}$. We henceforth assume that the factor graph $G_{\mathcal{F}}$ is connected, which implies also that the Markov graph $M_{\mathcal{F}}$ is connected. If $M_{\mathcal{F}}$ were not connected, then we could partition $\pindext$ as $\pindext=\pindex_1\cup\pindex_2$ with $\pindex_1$ and $\pindex_2$ disconnected in $M_{\mathcal{F}}$ and with $i_{\qoi}\in\pindex_1$. It is then easy to show that $\npset_{\pindex_1}\ind \npset_{\pindex_2}$, and hence $U^2_{\pindex_2} = 1$. In other words, $\npset_{\pindex_2}$ is irrelevant to $\qoi$, so if $M_{\mathcal{F}}$ were disconnected, we could replace $M_{\mathcal{F}}$ by the connected component of $i_{\qoi}$ and proceed without loss of generality.

A subset $\pindex'\subseteq \pindext$ separates disjoint sets $\pindex_1,\pindex_2\subseteq \pindext\setminus \pindex'$ if every path from $i\in\pindex_1$ to $j\in\pindex_2$ in $M_{\mathcal{F}}$ includes some $k\in\pindex'$. If $p(\nptotal\mid\data)$ factorizes over $\mathcal{F}$, then standard results for graphical models tell us that if $\pindex'$ separates $\pindex_1$ and $\pindex_2$ in $M_{\mathcal{F}}$, then $\npset_{\pindex_1}\ind\npset_{\pindex_2}\mid\npset_{\pindex'}$. In other words, we can identify conditional independence relationships in $p(\nptotal\mid\data)$ by locating separating sets in $M_{\mathcal{F}}$. We refer the reader to \cite{graphical_handbook} for details on graphical models.

This suggests a procedure for constructing a \DTree from a root node $\node_0$ and a set of leaf nodes $\{\nodel_i\}_{i=1}^l$. In particular, by iteratively searching $M_{\mathcal{F}}$ for separating sets, we can construct Markov chains starting at $\qoi=\npset(\node_0)$ and terminating at the leaf quantities $\npset(\nodel_i)$.
We first sketch this procedure for the first leaf $\nodel_1$, before giving a formal presentation in Algorithm \ref{alg:markov_chain}.
At each iteration $k\geq 1$, we start with a chain $\pichain=(\pindex_1,\pindex_2,\ldots,\pindex_{k+1})$ such that $\npset(\pichain) \stackrel{\mathrm{def}}{=} \left( \npii{1},\npii{2},\ldots,\npii{k+1} \right)$ satisfies the Markov condition \eqref{eq:markov_path_condition}, and such that $\npii{1}=\qoi$ and $\npii{k+1}=\psi(\nodel_1)$. We also specify an insertion point $1\leq j\leq k$. We then attempt to insert a set $\pindex'$ between $\pindex_j$ and $\pindex_{j+1}$, forming a new chain $\pichain'$ such that $\npset(\pichain')$ retains property \eqref{eq:markov_path_condition}.

To achieve this, we form $\pichain^l = \left( \pindex_1,\ldots,\pindex_j \right)$ and $\pichain^r = \left( \pindex_{j+1},\ldots,\pindex_{k+1} \right)$ and attempt to construct two separating sets $S^*_M(\pichain^l; \pichain^r)\subseteq \pindext$ and $S^*_M(\pichain^r; \pichain^l)\subseteq\pindext$ neighboring $\pindex^l = \bigcup_{i=1}^j\pindex_i$ and $\pindex^r = \bigcup_{i=j+1}^{k+1}\pindex_i$ respectively. If $S^*_M(\pichain^l; \pichain^r) = S^*_M(\pichain^r; \pichain^l)=\emptyset$, then no separator exists, and we terminate. 
If the separators are nonempty, we take $\pindex' = S^*_M(\pichain^l; \pichain^r)$ if $S_{\mathrm{lik}}(S^*_M(\pichain^l; \pichain^r)) \leq S_{\mathrm{lik}}(S^*_M(\pichain^r; \pichain^l))$, and $\pindex' = S^*_M(\pichain^r; \pichain^l)$ otherwise.
We then insert $\pindex'$ at index $j$, forming the new sequence $\pichain' = \left( \pindex_1,\ldots,\pindex_j,\pindex',\pindex_{j+1},\ldots,\pindex_{k+1} \right)$. Our insertion point at iteration $k+1$ is then $j+1$ or $j$ depending on whether $\pindex'$ neighbors $\pindex_{j}$ or $\pindex_{j+1}$.

\begin{algorithm}
  \begin{algorithmic} \onehalfspacing
  \Require Connected, undirected graph $M$, start and end sets $\pindex_s, \pindex_e\subseteq \pindext$, threshold $\gamma\in(0,1)$.
  \Function{MarkovChain}{$M$, $\pindex_s$, $\pindex_e$, $\gamma$}
  \State $\pichain \gets \left( \pindex_s, \pindex_e \right); \quad i^* \gets 1;$
  \State $F_1 \gets \pindex_s;\quad F_2\gets \pindex_e;$
  \While{$N_{M}(F_1)\cap F_2 = \emptyset$}
    \State $S_1 \gets S^*_M\left( F_1; F_2 \right);\quad S_2 \gets S^*_M\left( F_2; F_1 \right);$
    \State $C_1 \gets S_{\mathrm{lik}}(S_1);\quad C_2\gets S_{\mathrm{lik}}(S_2);$
    \If{$\min\{C_1, C_2\} > \gamma$}
      \State $\mathbf{break};$
    \ElsIf{$C_1 \leq C_2$}
      \State $S \gets S_1;\quad F_1 \gets F_1\cup S;$
      \State $i^*\gets i^* + 1;$
    \ElsIf{$C_1 > C_2$}
      \State $S \gets S_2;\quad F_2 \gets F_2 \cup S;$
    \EndIf
    \State $\pichain \gets \left( \pichain_{\leq i^*}, S, \pichain_{> i^*} \right);$
  \EndWhile
  \State \Output $\pichain$
  \EndFunction
  \end{algorithmic}
  \caption{Markov Chain Construction}
  \label{alg:markov_chain}
\end{algorithm}

This procedure is initialized with the sequence $\left( \pindex_1,\pindex_2 \right) = \left( \{i_{\qoi}\}, \pindex(\nodel_i) \right)$ and $j=1$. We terminate when we either cannot locate a separating set $\pindex'$, or when the likelihood specificity $S_{\mathrm{lik}}(\pindex')$ exceeds a user-specified threshold $\gamma\in (0,1)$. Algorithm \ref{alg:markov_chain} gives the full iteration. A formal definition of the separating set operator $S^*_M(\cdot;\cdot)$ as well as a proof that Algorithm \ref{alg:markov_chain} produces valid Markov chains is given in the supplementary material.

To combine the generated chains (one for each leaf node) into a tree, we need two more operations. First, we can trivially promote any single chain $\pichain=\left( \pindex_1,\pindex_2,\ldots\pindex_k \right)$ into a tree as follows:
\begin{equation}
  \label{eq:promotion_def}
  \mathrm{Tree}(\pichain) = \left(\left( \node_i \right)_{i=1}^{ k }, \left[\node_2 = \node_1,\node_3 = \node_2,\ldots, \node_{ k } = \node_{ k-1 } \right]\right),
\end{equation}
where $\pindex(\node_i )= \pindex_{i}$, and $\left[\node_2 = \node_1,\node_3 = \node_2,\ldots, \node_{ k } = \node_{ k-1 } \right]$ defines the parent mapping $P$ in Definition \ref{def:sensitivity_tree}.
Then, given a \DTree $T$ and chain $\pichain=\left( \pindex_1,\pindex_2,\ldots,\pindex_k \right)$ with $\pindex_1=\{i_{\qoi}\}$, we need a method of ``grafting'' $\pichain$ onto $T$.
\begin{definition}[Grafting]
  Let $\node^*$ be the terminal node of the longest rooted path $\left(\node_1,\node_2,\ldots,\node_m\right)$ in $T$ such that $\pindex(\node_i) = \pindex_i$ for all $1\leq i\leq m\leq k$. (The uniqueness of this longest path is guaranteed by the uniqueness of children in Definition \ref{def:sensitivity_tree}.) For $j = 1,\ldots, k-m$, define nodes $\node_j^{\mathrm{new}}$ with $\pindex(\node_j^{\mathrm{new}}) = \pindex_{m+j}$. Then the tree with $\pichain$ grafted is defined as
  \begin{equation}
    \label{eq:graft_def}
    \underset{\pichain}{\mathrm{Graft}}(T) = \left( \mathcal{N}_T\cup \{\node_j^{\mathrm{new}}\}_{j=1}^{ k-m }, P\left[\node_1^{\mathrm{new}} = \node^*, \node_2^{\mathrm{new}} = \node_1^{\mathrm{new}},\ldots, \node_{ k-m }^{\mathrm{new}} = \node_{ k-m-1 }^{\mathrm{new}}\right] \right),
  \end{equation}
  where $P[...]$ denotes the parent relation $P$ in $T$ with the additional relations $...$ added.
\end{definition}

\begin{algorithm}
  \begin{algorithmic} \onehalfspacing
  \Require Connected, undirected graph $M$, root index $i_{\qoi}$, leaf nodes $\{\nodel_i\}_{i=1}^l$, threshold $\gamma\in(0,1)$.
  \For{$i\in\{1,\ldots,l\}$}
    \State $\pichain_i\gets\mathrm{MarkovChain}\left( M, \pindex_s = \{i_{\qoi}\}, \pindex_e = \pindex(\nodel_i), \gamma \right);$
  \EndFor
  \State $T \gets \mathrm{Tree}\left( \pichain_1 \right);$
  \For{$j \in \{2,\ldots, l\}$}
    \State $T \gets \underset{\pichain_j}{\mathrm{Graft}}\left( T \right);$
  \EndFor
  \State \Output $T$
  \end{algorithmic}
  \caption{\DDTree Construction}
  \label{alg:sensitivity_tree_from_chains_simple}
\end{algorithm} 

With the promotion and grafting operations defined, we can build a tree out of a sequence of leaf nodes $\{\nodel_i\}_{i=1}^l$ using Algorithm \ref{alg:sensitivity_tree_from_chains_simple}. A proof that Algorithm \ref{alg:sensitivity_tree_from_chains_simple} yields valid \DTrees is given in the supplementary materials.

\subsection{Operations on \DDTreesNS}

Given an initial \DTree, we may want to add nodes to find the better explanatory sets. To this end, we develop a suite of operations which modify \DTreesNS.

\textbf{Branching from a node.}
Given a node $\node$ in $T$, we may ask if a proper subset $\pindex'\subset\pindex(\node)$ is almost as good in the sense that $U^2_{\pindex'}\approx U^2_{\pindex(\node)}$. We add $\pindex'$ to $T$ with a branching operation.

\begin{definition}[Branching]
  \label{def:branch}
  Let $T=\left( \mathcal{N}, P \right)$ be a \DTreeNS, $\node$ a node in $T$, and $\pindex'\subset\pindex(\node)$ a proper subset. Then we define the tree branched from $\pindex(\node)$ to $\pindex'$ as 
  \begin{equation}
    \label{eq:branch_def}
    \branch{\node}{\pindex'}{T}
    = \begin{cases}
      T, & \text{$\exists\node''\in T$ with $P(\node'')=\node$ and $\pindex(\node'')=\pindex'$}\\
      \left( \mathcal{N}\cup\{\node'\}, P\left[ \node' = \node \right]\right), & \text{otherwise},
    \end{cases}
  \end{equation}
  where $\node'$ is a new node with $\pindex(\node') = \pindex'$.
\end{definition}

\textbf{Subdividing a branch.} 
Consider a pair of nodes $\node_p$ and $\node_c$ with $P(\node_c) = \node_p$. If $U^2_{\pindex(\node_p)} \ll U^2_{\pindex(\node_c)}$, we may look for a third set that splits the difference between $\pindex(\npset_c)$ and $\pindex(\npset_p)$. This can be achieved by the following subdivision operation.

\begin{definition}[Subdividing]
  Let $T=\left( \mathcal{N}, P \right)$ be a \DTree and $\node_p, \node_c$ be nodes in $T$ with $P(\node_c) = \node_p$. For any $\pindex_p'\subseteq \pindex(\node_p)$, let $\pindex_s = \pindex_p'\cup\pindex(\node_c)$. We subdivide the branch from $\node_p$ to $\node_c$ by $\pindex_s$ as follows:
  \begin{equation}
    \label{eq:subdivide_def}
    \subdivide{\node_p}{\pindex_p'}{\node_c}{T} = \begin{cases}
      \left( \mathcal{N}, P[\node_c = \node']\right), & \text{$\exists\node'\in T$ with $P(\node') = \node_p$ and $\pindex(\node')=\pindex_s$}\\
      \left( \mathcal{N}\cup \{\node_s\}, P[\node_c = \node_s, \node_s = \node_p] \right), & \text{otherwise},
    \end{cases}
  \end{equation}
  where $\node_s$ is a new node with $\pindex(\node_s)=\pindex_s$.
\end{definition}

\textbf{Merging two nodes.}
For two nodes $\node_1$ and $\node_2$, it may be useful to consider the combination $\npset_{\pindex(\node_1)\cup\pindex(\node_2)}$. For example, let $\node^*$ be the deepest common ancestor of $\node_1$ and $\node_2$. If $\node^*$ is much more explanatory than $\node_1$ or $\node_2$ separately -- i.e. if $U^2_{\pindex(\node^*)} \ll U^2_{\pindex(\node_1)}$ and $U^2_{\pindex(\node^*)} \ll U^2_{\pindex(\node_2)}$ -- we may ask how much of this gap can be bridged by $\pindex(\node_1)\cup\pindex(\node_2)$. To place this union on a \DTree, we merge the corresponding nodes.

\begin{definition}[Merging]
  Let $T=\left( \mathcal{N}, P\right)$ be a \DTreeNS, and let $\node_1$ and $\node_2$ be nodes with deepest common ancestor $\node^*$. Let $\pichain_{\mathrm{merge}} = \mathrm{MarkovChain}\left( M, \pindex(\node^*), \pindex(\node_1)\cup\pindex(\node_2), \gamma \right)$. Then the tree with nodes $\node_1$ and $\node_2$ merged is given as 
  \begin{equation}
    \label{eq:merge_def}
    \merge{\node^*}{\node_1}{\node_2}{T} = 
    \branchplain{\node^m}{\pindex(\node_2)}\;\circ\;
    \branchplain{\node^m}{\pindex(\node_1)}\;\circ\;
    \underset{\pichain_{\mathrm{merge}}}{\mathrm{Graft}}(T),
  \end{equation}
  where $\circ$ denotes compoisition of functions, and $\node^m$ is the unique leaf node introduced in $\underset{\pichain_{\mathrm{merge}}}{\mathrm{Graft}}(T)$ with $\pindex(\node^m) = \pindex(\node_1)\cup\pindex(\node_2)$.
\end{definition}

Each of these operations preserve the invariants of \DTrees.
\begin{lemma}\label{lem:st_preserved}
  Let $p$ be a density which factorizes over $\mathcal{F}$, and let $M_{\mathcal{F}}$ be the corresponding Markov graph. Let $T=\left( \mathcal{N}, P \right)$ be a \DTree for $p$ and $M_{\mathcal{F}}$ generated by Algorithm \ref{alg:sensitivity_tree_from_chains_simple}. Let $\node_1$ and $\node_2$ be nodes in $T$.
  \begin{enumerate}
    \item If $\pindex'\subset \pindex(\node_1)$, then $\branch{\node_1}{\pindex'}{T}$ is a \DTreeNS.
    \item If $P(\node_2)=\node_1$ and $\pindex'\subseteq \pindex(\node_1)$, then $\subdivide{\node_1}{\pindex'}{\node_2}{T}$ is a \DTree.
    \item If $\node^*$ is the deepest common ancestor of $\node_1$ and $\node_2$, then $\merge{\node^*}{\node_1}{\node_2}{T}$ is a \DTree.
  \end{enumerate}
\end{lemma}

\begin{proof}
  See supplementary materials.
\end{proof}

\textit{Remark:} The statement of Lemma \ref{lem:st_preserved} requires that the \DTree be generated by Algorithm \ref{alg:sensitivity_tree_from_chains_simple} because this algorithm produces objects with slightly stronger independence properties than those guaranteed by the definition of \DTreesNS. A definition of these stronger objects, which we term \textit{separating} trees, is given in the supplementary material.

\section{Example: Causal Inference in Simulated Panel Data}\label{sec:ex1_sc}

Our goal in this example is inference for a treatment effect parameter in a simulated causal inference problem based on panel data. We simulate data $\data_{it}$ for $1\leq i\leq N=6$ units over $1\leq t\leq T=10$ time points, with unit 1 exposed to a treatment after time $t_{\delta}=6$. Figure \ref{fig:sc_data} plots the simulated data, with vertical lines showing the treatment time for unit 1. 

\begin{figure}
  \centering
  \includegraphics[width=0.66\textwidth]{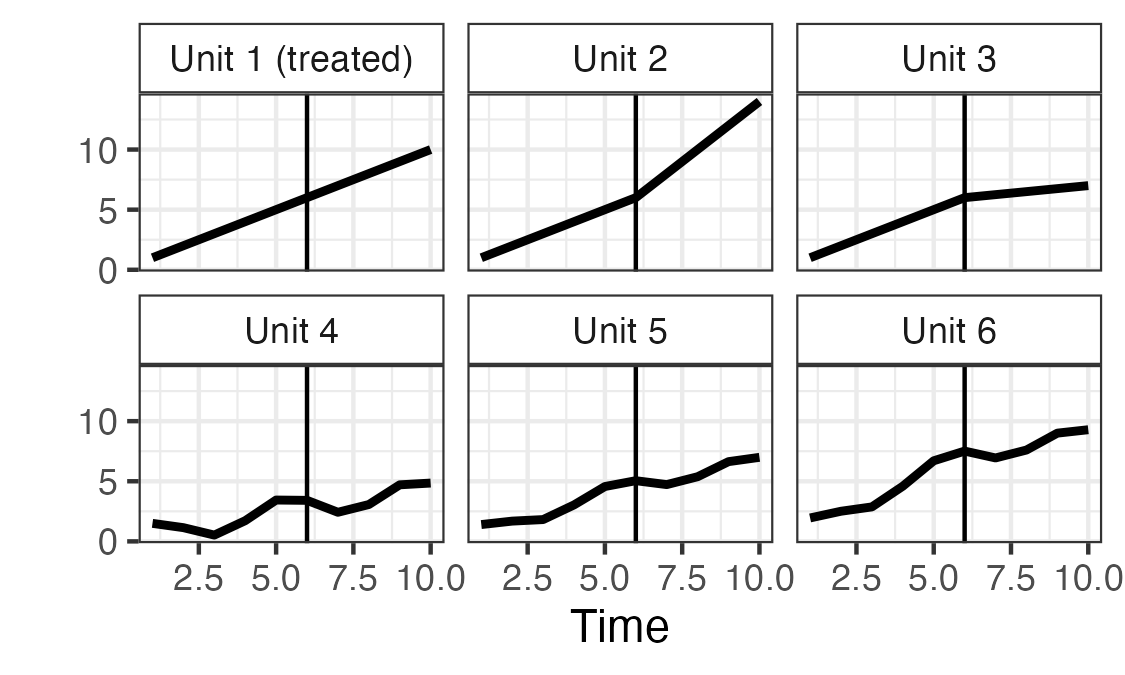}
  \caption{The simulated panel data. Units 1-3 are all identical in the pre-treatment period, but divergent in the post-treatment period. Only unit 1 is treated.}
  \label{fig:sc_data}
\end{figure}

Estimating treatment effects relies on inferring correlations between units in the pre-treatment period which can be extrapolated into the post-treatment period. In particular, we rely on the post-treatment behavior of the untreated units to learn how the treated unit was likely to behave in the absence of treatment. This strategy and the corresponding inference for the treatment effect can be undermined in multiple ways.

A lack of relevant data is the most straightforward threat to inference. For instance, in Figure \ref{fig:sc_data}, units 4 -- 6 are a poor match to unit 1 in the pre-treatment period, and thus provide little reliable information about unit 1 in the post-treatment period.
On the other hand, units 2 and 3 closely match unit 1 in the pre-treatment period, but create a different threat to inference: Despite matching unit 1 equally well before treatment, units 2 and 3 diverge and give incompatible predictions in the post-treatment period. In this case, we have too much potentially relevant information to infer the treatment effect with confidence.

\subsection{Parameters and Model}

We fit the data in Figure \ref{fig:sc_data} using a latent factor model augmented with a treatment effect parameter. Our model closely resembles Bayesian approaches to synthetic control, a causal inference methodology for observational panel data \citep{sc_abadie, bayes_sc}. Specifically, the matrix of observations $\datam \in\RR^{N\times T}$ is modeled as the sum of three terms:
\begin{equation}
  \label{eq:LFM}
  \datam = \lfload\lffac + \effect + \epsilonv.
\end{equation}
From right to left:
\begin{enumerate}
  \item $\epsilonv\in \mathbb{R}^{N\times T}$ is a matrix of independent, zero-mean Gaussian errors which are independent of all other quantities. The entries of $\epsilonv$ have constant variance $\sigma^2_i$ along each row $1\leq i\leq N$, i.e. $\mathrm{Var}(\epsilonc_{it})=\sigma^2_i$ for all $1\leq t\leq T$. 
  \item $\effect\in\mathbb{R}^{N\times T}$ is the matrix of treatment effects, i.e. $\effectc_{it} = \effects_{t - t_{\delta}} \mathbbm{1}_{\{i=1, t > t_{\delta}\}}$, where $\mathbbm{1}_{S}$ is the indicator of $S$, and $\delta\in \RR^{T-t_{\delta}}$ is a vector of nonzero treatment effects.
  \item $\lffac\in\mathbb{R}^{K\times T}$ represents $K$ time-varying latent factors, and $\lfload\in\mathbb{R}^{N\times K}$ contains factor loadings for each unit. The product of $\lfload$ and $\lffac$ is referred to as the latent component.
\end{enumerate}

The latent factors (i.e. rows of $\lffac$) are independent and assigned $\mathrm{AR}(1)$ priors with unit stationary variance. The loadings matrix $\lfload$ is assumed lower triangular with positive diagonal in order to eliminate nonidentification that otherwise arises from the product $\lfload\lffac$. This constraint on $\lfload$ is typical in latent factor models and, importantly, does not constrain the latent component $\lfload\lffac$ \citep{bayes_lfm_id}. 

The untreated potential outcome is defined as $\datam(0) = \datam - \effect$, i.e. what would be observed had the treatment not occurred. Conditional on $\lfload$ and the $\sigma_i^2$, the cross-sectional covariance of the untreated potential outcomes 
\[
\data_{\cdot t}(0) = (\data_{1t}(0), \data_{2t}(0), \ldots, \data_{Nt}(0))
\] 
is $\lfload\lfload^T+\diag\left( \sigma^2 \right)$ for all $1\leq t\leq T$, where $\diag\left(\sigma^2\right)$ is the $N\times N$ diagonal matrix with $\diag\left(\sigma^2\right)_{ii} = \sigma^2_i$. Therefore, the elements of $\lfload$ directly control the cross-sectional latent correlation between (untreated) units. Complete details of the prior specification can be found in the supplementary materials.

We fit the model \eqref{eq:LFM} to the simulated data in Figure \ref{fig:sc_data} using Stan \citep{stan}. As expected, we find weakly identified marginal posteriors for all effects $\delta_t$. For simplicity, we take the effect at the final time, $\delta_{4}$, as our quantity of interest. The posterior for $\delta_{4}$ is centered near $0$ with a $95\%$ credible interval of $[-3.95, 4.17]$. Over the pre-treatment period, the treated unit had an observed standard deviation of $1.9$, so essentially any plausible effect remains possible a posteriori.

\subsection{Analysis with \TTrees}

We consider two types of hypothetical data: (1) observations $\datapotm\in\RR^{N\times T^*}$ of each observed unit at $T^*=6$ additional pre-treatment times, and (2) observations $\datapot^{\mathrm{new}}\in \RR^{T+T^*}$ of a new unit at all times, resulting in $N^*=N+1$ units. The hypothetical data are sampled from their posterior predictive distribution (details given in the supplementary materials). Using Algorithm \ref{alg:sensitivity_tree_from_chains_simple}, we construct a \DTree rooted at $\delta_{4}$ with $N^*$ leaf nodes representing the hypothetical data for each unit. Figure \ref{fig:init_ex_ptree1} displays the corresponding \TreeNS.

\begin{figure}
  \centering
  \includegraphics[width=0.95\textwidth]{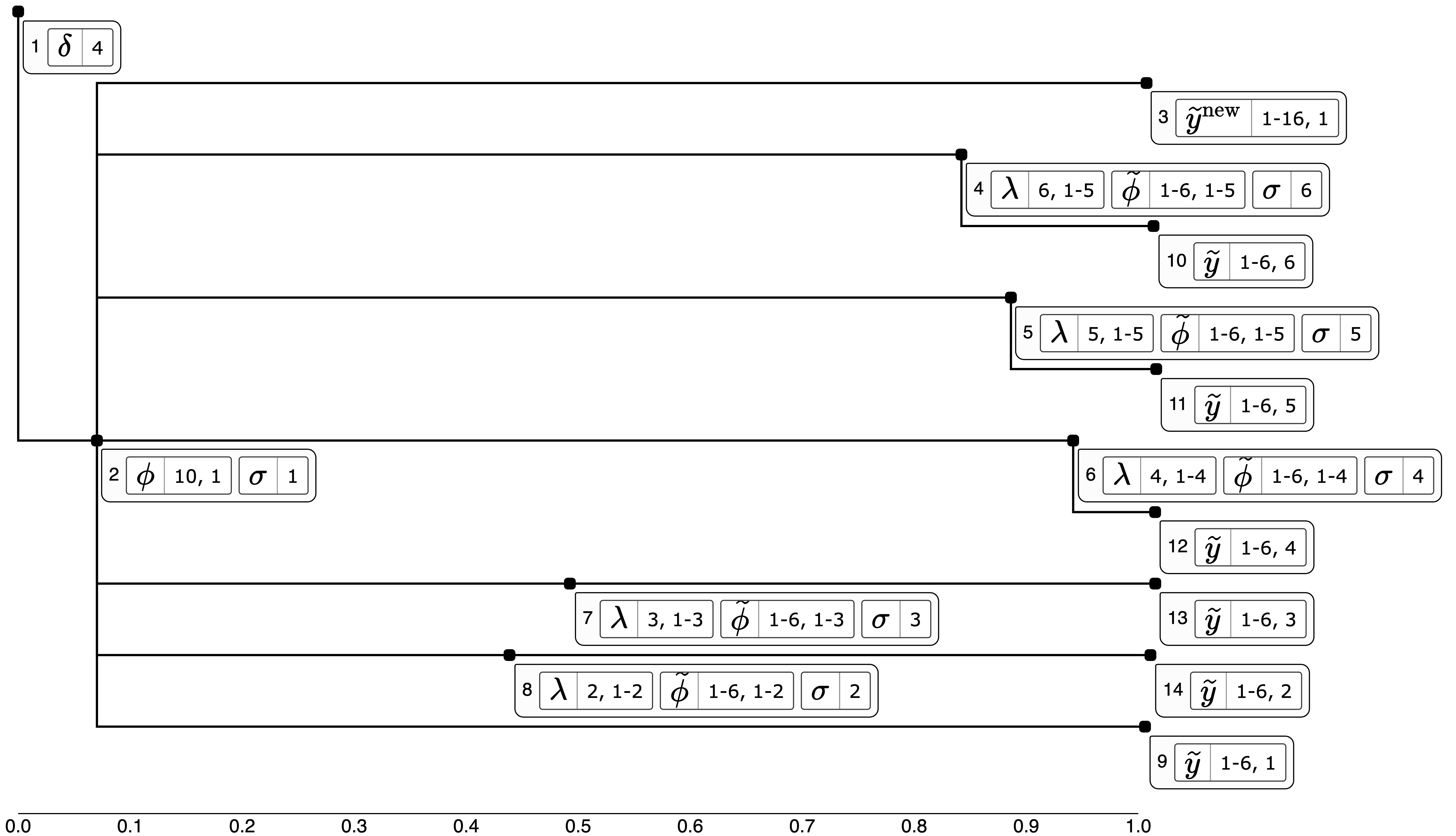}
  \caption{The generated \Tree with leaves for $\datapot_{i\cdot}$ ($1\leq i\leq N$) and $\datapot^{\mathrm{new}}$.}
  \label{fig:init_ex_ptree1}
\end{figure}

All leaf nodes have root uncertainty indices near $1$, indicating that observing the hypothetical data for any single unit is unlikely to improve the identification of $\delta_{4}$. Therefore, if we want to learn about $\effects_4$ using this hypothetical data, it will be necessary to combine hypothetical data across units. While this can be achieved with the merge operation defined in Section \ref{sec:post_trees}, we don't yet know which nodes to merge. 

\begin{figure}
  \centering
  \includegraphics[width=0.95\textwidth]{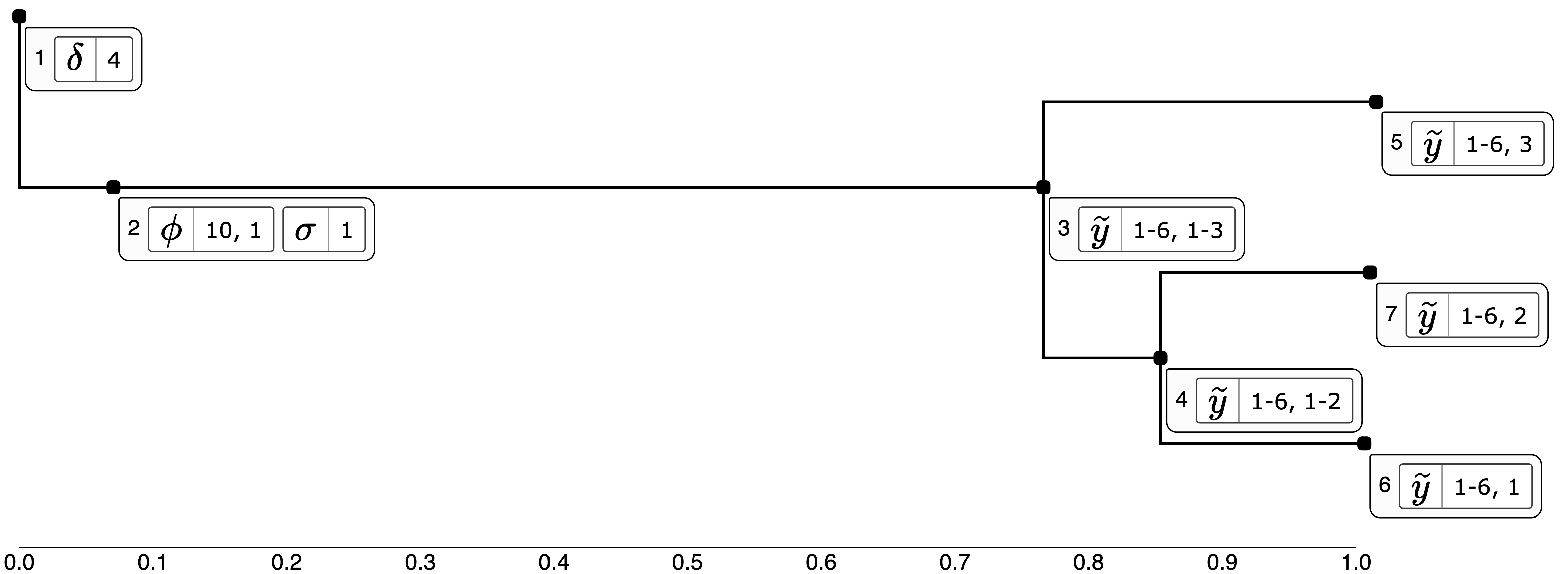}
  \caption{The \Tree for the leaves corresponding to units 1-3, with leaves 1 \& 2 merged (top plot) and all leaves merged (bottom plot).}
  \label{fig:init_ex_ptree_merge}
\end{figure}

Rather than manually guess, we can automatically search for the pair of leaves $\nodel_1$ and $\nodel_2$ such that $U^2_{\pindex(\nodel_1)\cup\pindex(\nodel_2)}$ is minimal. The \texttt{variance-deltas} package implements this automation, and the result of applying it twice is displayed in Figure \ref{fig:init_ex_ptree_merge}, showing only the newly added paths.
The greatest reduction in uncertainty comes from observing $\datapot_{1\cdot}$ and $\datapot_{2\cdot}$ together, with further improvement from adding $\datapot_{3\cdot}$. Specifically, combining $\datapot_{1\cdot}$, $\datapot_{2\cdot}$, and $\datapot_{3\cdot}$ yields expected posterior standard deviation $< 75\%$ of the original posterior standard deviation, at the cost of observing $30\%$ more data. 

Furthermore, in Figure \ref{fig:init_ex_ptree1}, the paths terminating in $\datapot_{2\cdot}$ and $\datapot_{3\cdot}$ (the top two paths) also included the best intermediate nodes. These intermediate nodes correspond to the loadings $\lfloadc_{2\cdot}$ and $\lfloadc_{3\cdot}$ respectively, along with the latent factors $\lffacnew$ for the new pre-treatment times. Since the loadings $\lfload$ control the latent correlations between units, our analysis so far suggests that the primary bottleneck to learning about $\effects_4$ is uncertainty about the latent correlations between unit 1 (the treated unit) and units 2 and 3.

To confirm this, we extend our tree again, merging the intermediate nodes and then branching to isolate the loadings. Figure \ref{fig:init_ex_ptree_inner} displays the path from this node to the root. As expected, exact knowledge of the loadings for units 2 and 3 reduces expected uncertainty to nearly a quarter of its current level.

\begin{figure}
  \centering
  \includegraphics[width=0.95\textwidth]{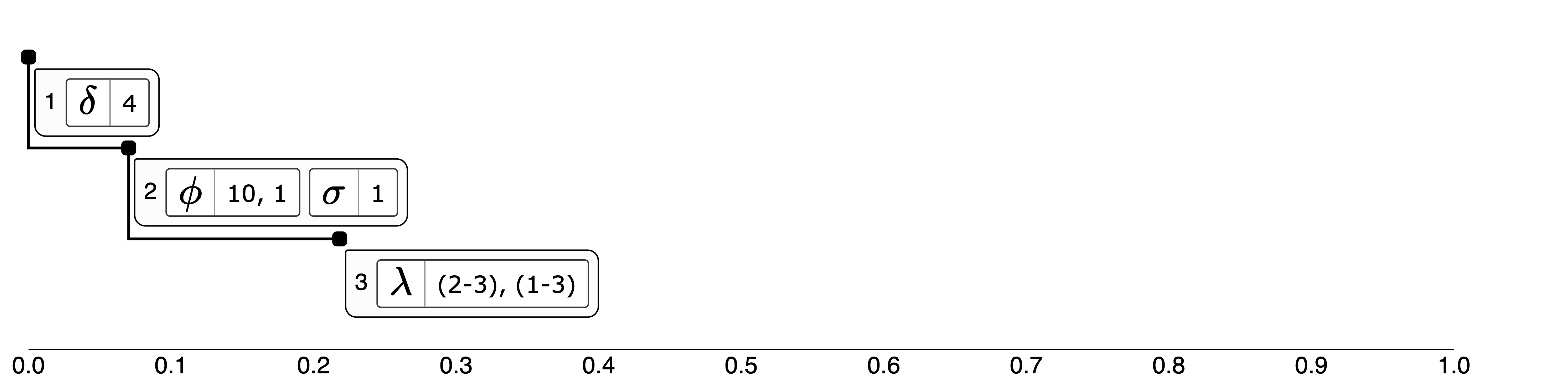}
  \caption{The path connecting the node for $\lfloadc_{2\cdot}$ and $\lfloadc_{3\cdot}$ to the root.}
  \label{fig:init_ex_ptree_inner}
\end{figure}

\section{Polling Model Example}\label{sec:ex2_polls}

Our second example applies \Trees to characterize uncertainty in a complex model of polls from the 2016 U.S. presidential election. In particular, we analyze the model of \citet{gelman_polling}, which includes components for temporal trends, interstate correlation, and several sources of potential bias/measurement error. We use \Trees to isolate two model components which together explain most of the posterior uncertainty and which allow us to determine the extent to which this uncertainty could be reduced.

The data are a subset of polls from the 2016 U.S. presidential election at the state and national level. Each poll reports the number of respondents and the number supporting the Democratic candidate, Hilary Clinton. Our goal is to forecast the proportion of voters who will support the Democratic candidate on election day in each state.

\subsection{Parameters and Model}

Let $1\leq t\leq T$ index the number of days in the model, with $t=1$ being the day of the first recorded poll, and $t=T$ being election day. The proportion of voters supporting the Democratic candidate in each state on each day is represented (on the logistic scale) by the matrix $\pmmean^b\in\RR^{51\times T}$. The columns of $\pmmean^b$ are assigned a time series prior:
\begin{equation}\label{eq:dem_support_prior}
  \pmmeanc^b_{\cdot t}\mid\pmmeanc^b_{\cdot (t+1)} \sim\mathrm{normal}\left( \pmmeanc^b_{\cdot (t+1)}, \mathbf{S}^{\pmmean} \right) \text{ for } 1\leq t\leq T-1, \text{ and } \pmmeanc^b_{\cdot T}\sim \mathrm{normal}\left( m^{\mathrm{f}},\mathbf{S}^{\mathrm{f}} \right).
\end{equation}
Here $\mathbf{S}^{\pmmean}\in\RR^{51\times 51}$ is a hyperparameter encoding correlation between states and variation over time. Likewise, $m^f\in\RR^{51}$ and $\mathbf{S}^f\in\RR^{51\times 51}$ are hyperparameters governing the prior on the election day outcome. 

Results of state and national polls are given a binomial model combining $\pmmean^b$ with terms for various types of polling bias. Letting $i=1,\ldots,N_{\mathrm{state}}$ index state polls and $y_i$ denote the number of respondents supporting the Democratic candidate out of $n_i$ respondents in poll $i$, we have
\begin{equation}
  \label{eq:poll_data_model}
  \datac_i \sim \mathrm{binomial}\left(\mathrm{logit}^{-1}\left( \pmmeanc^{b}_{s_i,t_i}+\biasp_i\right), n_i \right),
\end{equation}
where $s_i$, $t_i$ denote the state and day for poll $i$, and $\beta_i$ models various sources of bias. National polls are modeled similarly, except that state-level terms including $\pmmeanc_{st}$ are averaged with weights accounting for each state's share of the national vote in the previous election.

The bias terms $\biasp_i$ in \eqref{eq:poll_data_model}, also modeled on the logistic scale, are decomposed into several further terms which are designed to capture the following sources of polling bias:
\begin{itemize}
\item A pollster-specific ``house'' effect $\pmmeanc^c$, modeling the nonzero, temporally stable bias that each pollster exhibits toward one party or the other.
\item A poll mode effect $\pmmeanc^m$, reflecting the phenomenon that the averages of phone- and internet-based polls tend to differ.
\item A population effect $\pmmeanc^r$, modeling differences between registered and likely voter polls.
\item A time-varying partisan nonresponse effect $\epsilon_t$, modeling differential rates of response between members of the two parties. The $\epsilon_t$ are assigned an AR(1) prior. 
\item State-level measurement error $\xi$ accounting for aggregate differences in the accuracy of polls in different states. This is assigned a multivariate normal prior with covariance determined by state correlations in previous elections and demographics.
\item Poll-level measurement error $\zeta$, a catch-all term representing additional variance in the polls' sampling frames unaccounted for by the above terms. These terms are assigned different prior scales at the state and national levels.
\end{itemize}
The complete model specification is given in \citet{gelman_polling}. Code and 2016 polling data for the model can be accessed at \url{https://github.com/TheEconomist/us-potus-model}.
We fit this model to available 2016 data using Stan \citep{stan}. 

Uncertain inference for the election-day Democratic support parameter $\pmmeanc^*_s = \pmmeanc_{sT}$ is particularly troublesome in states $s$ with $\E\left[ \pmmeanc^*_s\mid\data \right]\approx 0.5$, i.e. those believed to be swing states. As an example, we analyze the state of Wisconsin (indexed by $s^*=49$), taking $\pmmeanc^*_{s^*}$ as our quantity of interest. Transforming to the proportion scale, the $95\%$ credible interval for $\mathrm{logit}^{-1}\left( \pmmeanc^*_{s^*} \right)$ is $[.478, 0.564]$, indicating failure to forecast the winner of Wisconsin. A summary of analyses for other swing states is given in the supplementary material.

\subsection{Analysis with \TTrees}

We begin by constructing hypothetical polls $\datapot$ from the most common pollster of Wisconsin, conducted one week before the election ($t=39$), which is after the last available polls in our data set. For each hypothetical poll $\datapot_j$, we choose one of 3 polling modes, 3 target populations, and whether to adjust the result for differential response rates across parties, for a total of $3\times 3\times 2=18$ simulated polls. 

The initial \DTree is constructed from a single leaf node combining all 18 hypothetical polls. Figure \ref{fig:poll_ex_ptree} displays this tree. We note two important observations:
\begin{enumerate}
  \item The node labelled 2 corresponds to $\pmmeanc^{b}_{\cdot 40}$ --- the logistic-transformed Democratic support across all states on the day \textit{after} the hypothetical polls are conducted. This node occurs with root uncertainty index $\approx 1/3$. Therefore, at least 1/3 of the uncertainty about $\pmmeanc^*_{s^*}$ is purely temporal, and would remain even if we could learn the Democratic support $\pmmeanc^{b}_{s^*, 39}$ in Wisconsin on the day of the polls exactly.
  \item However, all 18 hypothetical polls together (the last node in Figure \ref{fig:poll_ex_ptree}) have root uncertainty index $\approx 0.9$. Thus, observing these new hypothetical polls likely would not come close to fully informing us about even $\pmmeanc^{b}_{s^*,39}$.
\end{enumerate}

\begin{figure}
  \centering
  \includegraphics[width=1\textwidth]{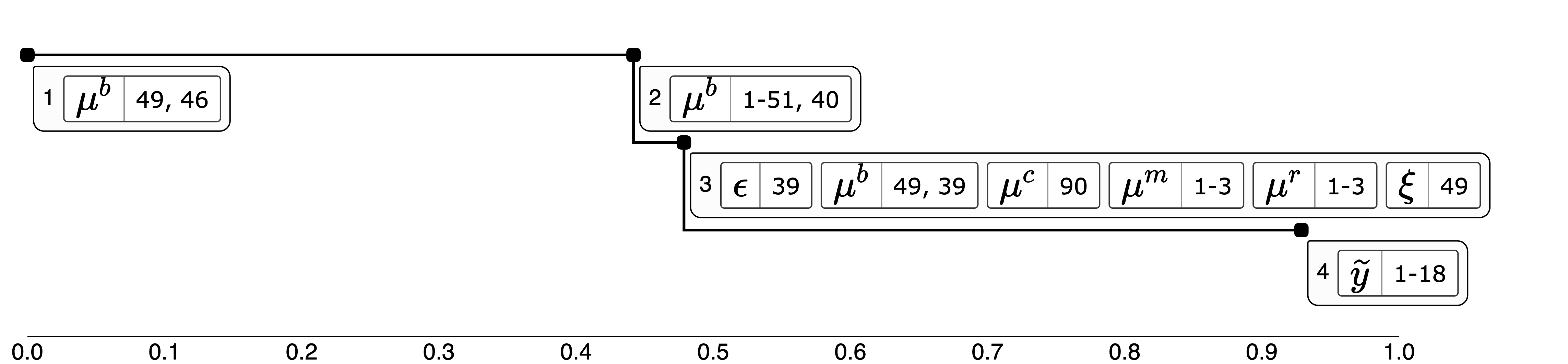}
  \caption{Initial \Tree with a single leaf representing all hypothetical polls.}
  \label{fig:poll_ex_ptree}
\end{figure}

The node containing $\pmmeanc^b_{s^*,39}$ (labelled $3$ in Figure \ref{fig:poll_ex_ptree}) also contains the bias parameters which affect the hypothetical polls. One or more these biases may thus explain the lack of information about $\pmmeanc^b_{s^*,39}$ in $\datapotset$. To test this, we insert a new node containing all bias parameters and $\datapotset$ (using the subdivision operation defined in Section \ref{sec:post_trees}). This new node represents how much would be learned from the hypothetical polls given exact knowledge of all biases.

Figure \ref{fig:poll_ex_ptree_sub} shows the \DTree obtained from this subdivision. The new node (labelled 4 in Figure \ref{fig:poll_ex_ptree_sub}) and its parent achieve essentially identical uncertainty indices. Since these nodes differ only in whether they include $\pmmeanc^b_{s^*, 38}$ or $\datapotset$, this implies that the lack of information in $\datapotset$ is due almost entirely to the bias terms.

To determine which of the bias parameters in nodes 3 and 4 are most important, we create new branches from node 4, one for each pair of $\datapotset$ and a single bias parameter. Figure \ref{fig:poll_ex_ptree_branched} displays the result. Only one bias parameter makes a substantial improvement in expected uncertainty when combined with $\datapotset$ --- the state-level polling bias $\xi$. Observing both $\datapotset$ and $\xi$ yields about 2/3 of the uncertainty reduction achievable by observing $\datapotset$ along with \textit{all} of the bias parameters, and about 1/2 of the uncertainty reduction achieved by exact knowledge of $\pmmeanc^{b}_{\cdot 40}$, the Democratic support in all states on the subsequent day.

\begin{figure}
  \centering
  \includegraphics[width=1\textwidth]{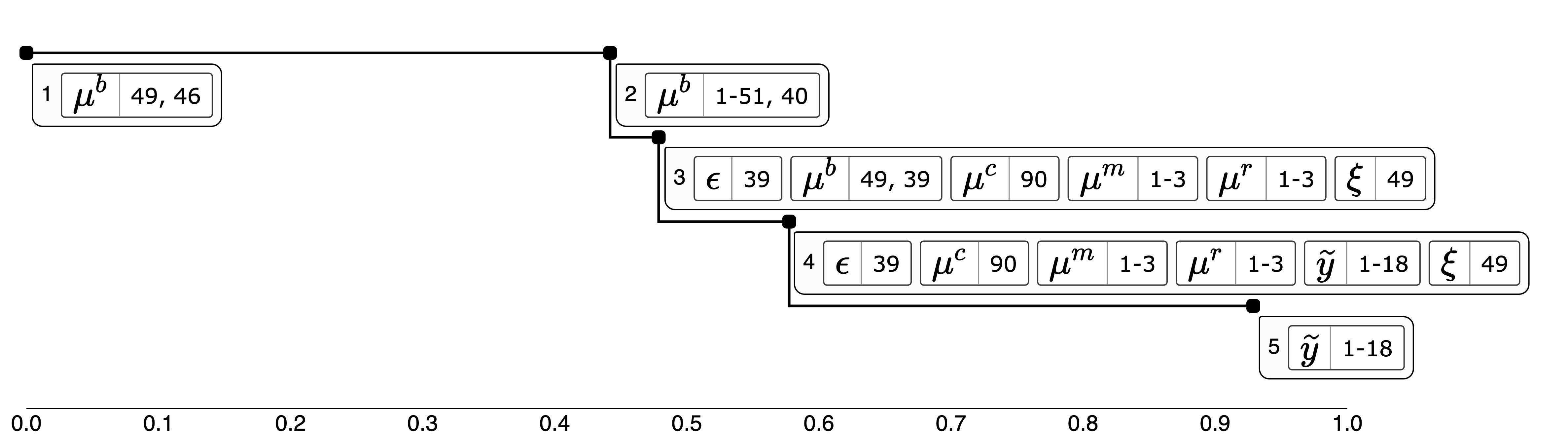}
  \caption{The result of subdividing the initial \Tree, adding a new node that combines the additional polls $\datapotset$ with all corresponding bias parameters.}
  \label{fig:poll_ex_ptree_sub}
\end{figure}

\begin{figure}[h]
  \centering
  \includegraphics[width=1\textwidth]{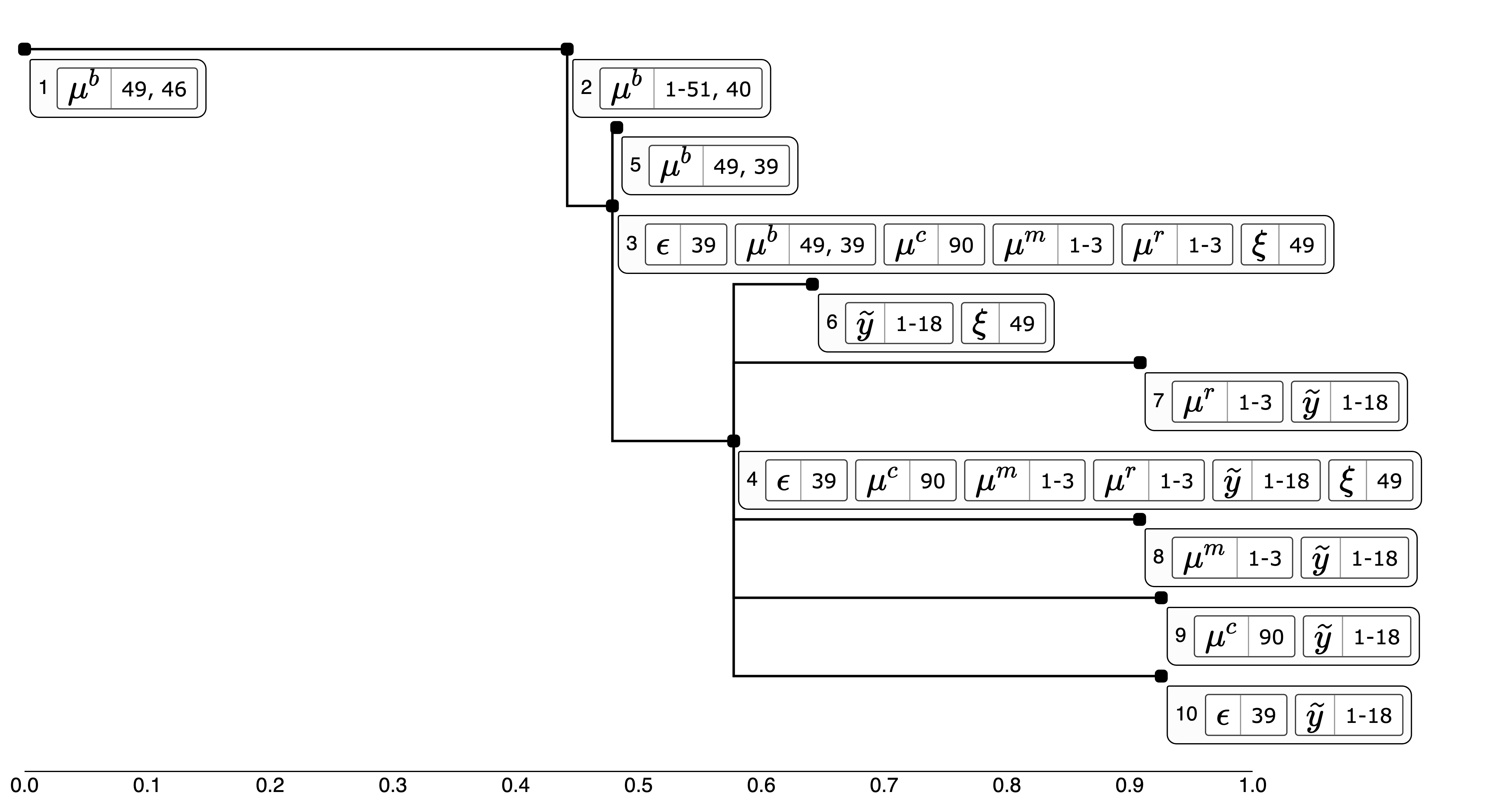}
  \caption{The result of adding branches from node 7 for each combination of additional polls $\datapotset$ and one of the bias parameters.}
  \label{fig:poll_ex_ptree_branched}
\end{figure}

In summary, we are able to draw several meaningful conclusions:
\begin{enumerate}
  \item Absent additional information, conducting more polls is unlikely to yield much improvement in the election day forecast for Wisconsin.
  \item Additional polling \textit{could} substantially reduce our uncertainty in combination with additional information about the state-level measurement error $\xi$.
  \item Even with this information, there is a nonnegligible limit to how much we can learn about the election day Democratic support in Wisconsin from polls conducted a week earlier, due to extrapolation uncertainty.
  \item Quantitatively, about 2/3 of the expected remaining uncertainty about $\pmmeanc^*_{s^*}$ given $\datapotset$ is explained by uncertainty about $\xi$ and extrapolation uncertainty.
\end{enumerate}

\section{Conclusion}\label{sec:conclusion}

We introduce \TreesNS, an interactive visualization method for explaining posterior uncertainty about scalar quantities of interest in Bayesian models, and implement this system in a software package \texttt{variance-deltas}. When the Bayesian model has a factorization $\mathcal{F}$, we show that the corresponding Markov graph $M_{\mathcal{F}}$ can be utilized to semi-automatically construct and transform \TreesNS, enabling the efficient generation and exploration of potential explanations.

Our two examples demonstrate that \Trees are capable of rapidly eliminating poor explanations of uncertainty and focusing our attention on specific, interpretable subsets of model quantities which can explain a large fraction of our overall uncertainty. In the simulated panel data example, we are able to identify small combinations of units which would be useful to observe over a longer period, while determining that other units are of vanishing small value. In the polling model example, we are able to isolate a pair of parameters that explain most of our uncertainty, directing us towards particular types of information which are needed to remove bottlenecks from our analysis.

The usefulness and general applicability of \Trees is limited, however, by the structure of the graph $M_{\mathcal{F}}$. In particular, our ability to discover useful and specific explanations depends on $M_{\mathcal{F}}$ being not too dense. When $M_{\mathcal{F}}$ is \textit{almost} sparse -- in that a few vertices have many neighbors -- we can recover a sparse graph by conditioning on these vertices at every step. (Details of this adjustment are given in the supplementary materials.) But when the graph $M_{\mathcal{F}}$ is very far from sparse, as can occur for instance in large regressions, it may be impossible to construct useful \TreesNS, or the explanatory sets $\pindex$ may be too large to reasonably estimate the uncertainty index $U^2_{\pindex}$.

\section*{Acknowledgements}
The AI tool Claude Code (version 2.0.76) was used for the following tasks during the preparation of this manuscript: (1) identification of some related literature, (2) generation of small amounts of code, and (3) review of mathematical arguments.

AI tools were \textit{not} used to generate any of the text of this manuscript (including references and supplementary material). AI-generated code was manually reviewed line by line, and all publications identified by AI were manually verified as existing and appropriate references for the present work.

\section*{Declaration of Interest Statement}
\textit{The authors report there are no competing interests to declare.}

\bibliography{bibliography.bib}

\newpage

\section{Supplement}

\subsection{Software Package and Replication Data}

The \texttt{variance-deltas} software package may be accessed at \url{https://github.com/collin-cademartori/variance-deltas/tree/main}.
Binary distributions of stable releases are provided for many major platforms (including Windows, MacOS, and Linux). Please report any suspected bugs by opening an issue on the linked GitHub repository.

Scripts and data to reproduce the \Trees from the three examples in this paper are openly available at the Harvard Dataverse \url{https://doi.org/doi:10.7910/DVN/FVHLGG}.

\subsection{Separating Set Construction}

First, we review some basic terminology from graph theory and establish notation. An undirected graph $M = (V_M, E_M)$ consists of set of vertices $V_M$ and edges
\begin{equation}
  \label{eq:undirected_edges}
  E_M\subseteq\left\lbrace \{i, j\}\mid i,j\in V_M\right\rbrace.
\end{equation} 
A path in $M$ is any sequence of vertices $(i_1,i_2,\ldots, i_m)\in V_M^m$ such that $\{i_k, i_{k+1}\}\in E_M$ for all $1\leq k\leq m-1$ and any $m\geq 2$. For any set of vertices $F\subseteq V_M$, we define the following:
\begin{enumerate}
  \item $N_M(F)$ is the set of neighbors of $F$ in the graph $M$, given as
    \begin{equation}
      \label{eq:neighbors_set}
      N_M(F) = \left\lbrace j\in V_M\setminus F \;\Big\vert\; \{i, j\} \in E_M \text{ for some } i\in F\right\rbrace
    \end{equation}
    If $i\in V_M$ is a single vertex, then we will write $N_M(i)$ for $N_M(\{i\})$.
  \item $C_M(F)$ is the connected component of $F$ in $M$, i.e. the set of all $j\in V_M$ reachable by a path starting at some $i\in F$. This may be given explicitly as
  \begin{equation}
    \label{eq:connected_component}
    C_M(F) = \bigcup_{k=0}^{\infty} N^k_M(F),
  \end{equation}
  where $N^k_M(F) = \underbrace{N_M\circ N_M\circ\cdots \circ N_M}_{k\text{ times}}(F)$ for $k\geq 1$, and $N^0_M(F) = F$.
  \item $M - F$ is the graph with vertex set $V_{M-F}= V_M\setminus F$ and edge set
  \begin{equation}
    \label{eq:subtraction_edge_set}
    E_{M-F} = \left\lbrace \{i,j\}\in E_{M} \;\Big\vert\; i,j\in V_M\setminus F \right\rbrace.
  \end{equation}
\end{enumerate}  

Recall that a density $p(x)$ which factorizes over $\mathcal{F}$ has the useful property that conditional independence relationships between components of $x$ can be described by separating sets.
\begin{lemma}\label{lem:mrf_independence}
  Let $p(x)$ with $x\in\RR^{\pdim}$ factorize over $\mathcal{F}$, and let $M_{\mathcal{F}}$ be the corresponding Markov graph. Let $F_1$, $F_2$, and $S$ be subsets of $\{1,\ldots,\pdim\}$. Suppose that every path from a vertex $i\in F_1$ to a vertex $j\in F_2$ includes a vertex in $S$. Then we have that
  \begin{equation}
    \label{eq:global_markov_property}
    x_{F_1} \ind x_{F_2} \mid x_S.
  \end{equation}
  In this case, the set $S$ is called a \textbf{separating set} for $F_1$ and $F_2$. If there exists any separating set for $F_1$ and $F_2$, then these sets are said to be \textbf{separable}.
\end{lemma}

A separating set $S$ for vertex sets $F_1$ and $F_2$ is called a minimal separating set if no proper subset of $S$ is a separating set. For separable $F_1$ and $F_2$, there may be many minimal separating sets. The following lemma gives one method to find a minimal separator.

\begin{lemma}\label{lem:minimal_separator}
  Let $F_1$ and $F_2$ be separable vertex sets in a connected, undirected graph $M$. Then the following are minimal separators of $F_1$ and $F_2$:
  \begin{equation}
    \label{eq:minimal_separators}
    S_M(F_1; F_2) = N_M(G_1) \cap N_{M}(F_1), \quad S_M(F_2; F_1) = N_M(G_2) \cap N_{M}(F_2),
  \end{equation}
  where $G_1 = C_{M-N_M(F_1)}(F_2)$, and $G_2 = C_{M-N_M(F_2)}(F_1)$.
\end{lemma}

If $M$ is connected, then Lemma \ref{lem:minimal_separator} implies a simple algorithm for computing $S_M(F_1; F_2)$ and $S_M(F_2; F_1)$. To calculate $G_1$, for instance, we can perform a depth-first search starting at any node in $F_2$, collecting all discovered vertices while ignoring all edges with a vertex in $N_M(F_1)$. Then calculating $S_M(F_1; F_2)$ from $G_1$ reduces to trivial set operations, and $S_M(F_2; F_1)$ may be computed similarly.

\begin{proof}
  We prove that $S_M(F_1;F_2)$ is a minimal separator of $F_1$ and $F_2$. The proof that $S_M(F_2;F_1)$ is a minimal separator is entirely symmetric.

  We begin by showing that $S_M(F_1;F_2)$ is a separator for $F_1$ and $F_2$. Let $p$ be any path of length $m$ from a vertex in $F_1$ to a vertex in $F_2$.
  First, we select the last vertex in $p$ which lies in $N_M(F_1)$ Specifically, define
  \begin{equation}
    \label{eq:closest_to_f1}
    \begin{split}
        I_p &= \left\lbrace 1\leq k\leq m\;\Big\vert\; p_k\in N_M(F_1) \right\rbrace,\\
        i_* &= \max I_p.
    \end{split}
  \end{equation}
  We now claim that $I_p\neq\emptyset$ and $i_*< m$. To show $I_p\neq\emptyset$, we also define $J_p =\{1\leq k\leq m\mid p_k\notin F_1\}$ and $j_*=\min J_p$. Clearly $J_p\neq \emptyset$ since $p_{m}\in F_2$ by definition, and $F_1\cap F_2=\emptyset$ since $F_1$ and $F_2$ are assumed separable. Furthermore, $j_*>1$ since $p_1\in F_1$ by definition. Thus, we must have $p_{j_*-1}\in F_1$. Since $\{p_{j_*-1},P_{j_*}\}\in E_M$, it follows that $p_{j_*}\in N_M(F_1)$, and thus $j_*\in I_p$, so $I_p\neq \emptyset$. 
  
  Furthermore, if $m\in Ip$, then there would exist some $i'\in F_1$ with $(i', p_m)\in E_M$. But then $(i', p_m)$ would be a path from $F_1$ to $F_2$, contradicting the assumption that $F_1$ and $F_2$ are separable. Therefore, $i^*< m$, as claimed.

  Now consider the truncated path $p' = (p_{i_*+1}, \ldots, p_{ m })$. Note that $p_{m }\in F_2$, and that $p'_k\notin N_M(F_1)$ for all $1\leq k\leq |p'|$. If this were not true, there would exist a $k_*>i_*$ such that $p_{k_*}\in N_M(F_1)$, which contradicts the definition of $i_*$. In other words, $p'_1$ is connected to $F_2$ by a path which does not pass through $N_M(F_1)$, and thus $p'_1\in C_{M-N_M(F_1)}(F_2)$. Since $p'_1 = p_{i_*+1}$, this implies that $p_{i_*}\in N_M(C_{M-N_M(F_1)}(F_2))$.

  Putting these pieces together, we conclude that
  \begin{equation}
    p_{i_*} \in N_M(F_1)\cap N_M(C_{M-N_M(F_1)}(F_2)) = N_M(F_1)\cap N_M(G_1),
  \end{equation}
  and therefore, since $p$ was an arbitrary path from $F_1$ to $F_2$, $N_M(F_1)\cap N_M(G_1)$ must separate $F_1$ and $F_2$.

  To show that this set is also a minimal separator, let $H\subset N_M(F_1)\cap N_M(G_1)$. Then let $j'\in (N_M(F_1)\cap N_M(G_1))\setminus H$. Because $j'\in N_M(F_1)$, there exists $i'\in F_1$ such that $\{i',j'\}\in E_M$. Note also that $i'\notin H$ since $i'\notin N_M(F_1)$. Furthermore, since $j'\in N_M(C_{M-N_M(F_1)}(F_2))$, there exists a path $p'$ (of length $m'$) such that $\{j', p'_1\}\in E_M$, $p'_m\in F_2$, and $p'_k\notin N_M(F_1)$ for all $1\leq k\leq m'$. The latter property implies in particular that $p'_k\notin H$ for all $1\leq k\leq m$. Putting this all together, $(i', j', p'_1,\ldots p'_m)$ is a path from $F_1$ to $F_2$ which does not include any vertex in $H$, showing that $H$ is not a separator for $F_1$ and $F_2$. Since $H$ was an arbitrary subset of $N_M(F_1)\cap N_M(G_1)$, we conclude that this set is in fact a minimal separator of $F_1$ and $F_2$, as claimed.
\end{proof}

Finally, the separating set function $S^*_M(F_1; F_2)$ used in Algorithm 2 is produced from $S_M(F_1; F_2)$ as
\begin{equation*}
  S^*_M(F_1; F_2) = \begin{cases}
    S_M(F_1; F_2), & F_1,F_2 \text{ separable},\\
    \emptyset, & \text{otherwise},
  \end{cases}
\end{equation*}
and similarly for $S^*_M(F_2; F_1)$.

\subsection{Proof that Algorithm \ref{alg:markov_chain} Produces Markov Chains}\label{sec:proof_alg_markov}

Rather than directly prove that Algorithm \ref{alg:markov_chain} yields valid Markov chains, we prove a slightly stronger result. We start with a definition.

\begin{definition}[Separating paths and trees]
  \label{def:sep_tree}
  Let $\pichain = \left( \pindex_1,\pindex_2\,\ldots,\pindex_k \right)$ be such that $\pindex_i\subseteq\pindext$ for all $1\leq i\leq k$. Let $M$ be a connected, undirected graph with $V_M = \pindext$. Then $\pichain$ is a \textbf{separating chain} with respect to $M$ if, for all $1\leq j\leq k$,
  \begin{equation}
    \label{eq:separting_tree_constraint}
    C_{M-\pindex_j}\left(\pindex_{\leq j}\setminus \pindex_j\right) \cap (\pindex_{\geq j}\setminus \pindex_j) = \emptyset,
  \end{equation}
  where $\pindex_{\leq j} = \cup_{i=1}^{j}\pindex_i$ and $\pindex_{\geq j} = \cup_{i=j}^k \pindex_i$. In words, each component $\pindex_j$ of $\pichain$ must separate its ancestors and descendants into distinct connected components in $M$.

  Let $\mathcal{F}$ be a factorization with connected Markov graph $M_{\mathcal{F}}$. Let $p$ factorize over $\mathcal{F}$, and let $T$ be a \DTree for $p$.
  If, for all paths $\left( \node_1,\ldots, \node_k \right)$ in $T$, $\left( \pindex(\node_1),\pindex(\node_2),\ldots,\pindex(\node_k) \right)$ is a separating chain with respect to $M_{\mathcal{F}}$, then $T$ is a \textbf{separating tree} for $p$ and $\mathcal{F}$.
\end{definition}

For $T$ to be a separating tree, it is sufficient that the property \eqref{eq:separting_tree_constraint} holds for all maximal paths in $T$.
We note that, in light of Lemma \ref{lem:mrf_independence}, the separating property \eqref{eq:separting_tree_constraint} implies the Markov property \eqref{eq:markov_path_condition} required for \DTreesNS.
We prove that Algorithm \ref{alg:markov_chain} produces valid separating chains, and it will follow immediately from this that it also produces valid Markov chains. First, we need the following lemma.

\begin{lemma}
  \label{lem:s_paths}
  Let $M$ be a connected, undirected graph with $V_M=\pindext$.
  For any separable $\pindex_1, \pindex_2\subset \pindext$, let $\pindex_S$ denote either $\minsep{\pindex_1}{\pindex_2}$ or $\minsep{\pindex_2}{\pindex_1}$. Then for any vertex $i\in\pindex_S$, there exist paths $p^1$ and $p^2$ connecting $i$ to $\pindex_1$ and $\pindex_2$ respectively, such that $p^1_j\notin \pindex_2$ and $p^2_k\notin \pindex_1$ for all $1\leq j\leq |p^1|$ and $1\leq k\leq |p^2|$.
\end{lemma}

\begin{proof}
  We prove the lemma for $\pindex_S=\minsep{\pindex_1}{\pindex_2}$, as the proof for the other case is entirely symmetric. Let $i\in \pindex_S$. Because $i\in N_M\left( \pindex_1 \right)$, there exists some $i'\in \pindex_1$ such that $\{i, i'\}\in E_{M}$. In fact, since $(\pindex_1\cup\pindex_S)\cap\pindex_2=\emptyset$, we have $\{i, i'\}\in E_{M-\pindex_2}$. Thus, $p^1 = (i, i')$ is a path from $i$ to $\pindex_1$ in $M-\pindex_2$.

  We also have $i\in N_M\left( C_{M-N_M(\pindex_1)}(\pindex_2) \right)$, so there exists some $i''\in  C_{M-N_M(\pindex_1)}(\pindex_2)\subseteq M-\pindex_1$ such that $\{i, i''\}\in E_{M-\pindex_1}$. Furthermore, by definition of the connected component, there exists some path $\widetilde{p}$ of length $k$ with $\widetilde{p}_1 = i''$, $\widetilde{p}_k\in \pindex_2$, and $\widetilde{p}_j\notin \pindex_1$ for all $1\leq j\leq k$ (since $\pindex_1$ and $\pindex_2$ are disconnected in $M-N_M(\pindex_1)$). Letting $\circ$ denote the composition of paths, it follows that $p^2 = (i, i'') \circ \widetilde{p}$ is a path in $M-\pindex_1$ connecting $i$ to $\pindex_2$, as required. 
\end{proof}

With this, we can now proceed with the proof that Algorithm \ref{alg:markov_chain} produces separating chains.

\begin{lemma}\label{lem:alg1_okay}
  The output of algorithm \ref{alg:markov_chain} is always a separating chain.
\end{lemma}

\begin{proof}
  We prove the slightly stronger claim that Algorithm \ref{alg:markov_chain} produces valid separating chains at the end of each iteration, proceeding by induction on the iteration.
  In the base case (i.e. first iteration), $\pichain=\left( \pindex_s,\pindex_e \right)$, which trivially satisfies \eqref{eq:separting_tree_constraint}. 
  
  Next, assume that at the end of the iteration $k\geq 2$, we have valid separating chain $\pichain'=\left( \pindex_1,\pindex_2,\ldots,\pindex_k \right)$. Let $i^{*}\in\{1,\ldots,k\}$ be the insertion point after iteration $k$, and define $\pindex^l=\cup_{i=1}^{i^*}\pindex_i$ and $\pindex^r = \cup_{i=i^*+1}^k\pindex_i$. By construction $F_1 = \pindex^l$ and $F_2 = \pindex^r$. If $\pindex^l$ and $\pindex^r$ are not separable, then Algorithm \ref{alg:markov_chain} outputs $\pichain'$, and we are done.
  
  If $\pindex^l$ and $\pindex^r$ are separable, let $\pichain$ be the chain resulting from applying one iteration of the while loop in Algorithm \ref{alg:markov_chain} to $\pichain'$, $F_1=\pindex^l$, $F_2=\pindex^r$, and $i^*$. Then it is easy to see that
  \begin{equation}
    \pichain = \left( \pichain'_1,\ldots, \pichain'_{i^*}, \pindex_S, \pichain'_{i^*+1}, \ldots, \pichain'_{k} \right),
  \end{equation}
  where $\pindex_S$ is either $\minsep{\pindex^l}{\pindex^r}$ or $\minsep{\pindex^r}{\pindex^l}$. For $j=i^*+1$, it is clear that $\pichain$ satisfies \eqref{eq:separting_tree_constraint} by Lemma \ref{lem:minimal_separator}.

  Next, suppose that $j\leq i^*$. Let $\pindex'_S=\pichain_j=\pichain'_j$. Then we must show that $\pindex'_S$ separates $\pichain_{<j}=\pichain'_{<j}$ and $\pichain_{>j}=\pichain'_{>j}\cup \pindex_S$. Clearly, $\pindex'_S$ separates $\pichain'_{<j}$ and $\pichain'_{>j}$ since $\pichain'$ already satisfies \eqref{eq:separting_tree_constraint}. Thus, it suffices to show that there is no path from $\pichain_{<j}$ to $\pindex_S$ that does not pass through $\pindex'_S$.

  If $\pichain_{<j}=\emptyset$, we are done. Otherwise, we demonstrate a contradiction. Suppose that there exists such a path $p$ originating at a vertex $j_1\in \pichain_{<j}=\pichain'_{<j}$ and terminating in some vertex $j_2\in\pindex_S$ with no vertices in $\pindex'_S$. Then by Lemma \ref{lem:s_paths} above, it follows that there exists another path $p'$ originating at $j_2$ and terminating in $\pichain'_{> i^*} \subseteq \pichain'_{> j}$ such that $p'$ does not pass through $\pichain'_{\leq i^*}\supset \pindex'_S$. It immediately follows that $p\circ p'$ is a path originating in $\pichain'_{<j}$ and terminating in $\pichain'_{> j}$ which does not pass through $\pindex'_S=\pichain'_j$, where $\circ$ denotes the composition of paths. But this contradicts the assumption that $\pichain'$ satisfied the separating condition \eqref{eq:separting_tree_constraint} for all $1\leq j\leq k$.
  
   Therefore, there are no paths from $\pichain_{<j}$ to $\pindex_S$ which do not pass through $\pindex'_S$, and thus $\pindex_j$ separates $\pichain_{< j}$ and $\pichain_{>j}$. The proof for $j>i^*+1$ follows in exactly the same fashion, and we conclude that $\pichain$ is again a separating chain. It then follows by induction that Algorithm \ref{alg:markov_chain} produces valid separating chains.
\end{proof}

\subsection{Proof that Algorithm \ref{alg:sensitivity_tree_from_chains_simple} Produces Valid \DDTreesNS}

As in the previous section, we prove a slightly stronger result -- that Algorithm \ref{alg:sensitivity_tree_from_chains_simple} produces valid separating trees. For this, it suffices to show the following two lemmas.

\begin{lemma}\label{lem:tree_promotion_valid}
  For every separating chain $\pichain$ starting at $\{i_{\qoi}\}$, $\mathrm{Tree}(\pichain)$ is a valid separating tree.
\end{lemma}

\begin{proof}
  To show that $\mathrm{Tree}(\pichain)$ is a separating tree, we have to verify the constraints in Definition \ref{def:sensitivity_tree}, as well as the separating condition \eqref{eq:separting_tree_constraint}. Let $\node_0$ be the root node of $\mathrm{Tree}(\pichain)$. By construction, we have that $\npset(\node_0) = \qoi$, so the first constraint is satisfied. The second constraint of Definition \ref{def:sensitivity_tree} follows from the separating tree condition \ref{eq:separting_tree_constraint}, and this is immediately satisfied for  $\mathrm{Tree}(\pichain)$ since $\pichain$ is a separating chain by assumption. Constraint 3 is trivially satisfied since each node has at most one child node.
\end{proof}

\begin{lemma}\label{lem:grafting_valid}
  For every separating chain $\pichain$ with $\pichain_1=\{i_{\qoi}\}$ and every separating tree $T$ rooted at $\qoi$, $\underset{\pichain}{\mathrm{Graft}}(T)$ is a valid separating tree.
\end{lemma}

\begin{proof}
We first note that the graft operation cannot modify the root node of $T$ since any new nodes are added as children of existing nodes. Since $T$ is already a valid \DTree for $\qoi$, the first constraint is thus satisfied. 

Next suppose that $\nodepath = (\node_1,\ldots,\node_k)$ is a path in $\underset{\pichain}{\mathrm{Graft}}(T)$. If $\nodepath_j\in T$ for all $1\leq j\leq k$, then $\nodepath$ is also a path in $T$ (since all nodes are in $T$ and grafting does not remove any branches from $T$). Because $T$ is a valid separating tree, $\nodepath$ must satisfy the condition \eqref{eq:separting_tree_constraint}. Now suppose that $\nodepath_j$ is not a node in $T$ for some $j\geq 1$. Then it follows from the definition of $\underset{\pichain}{\mathrm{Graft}}(T)$ that $(\pindex(\nodepath_1),\ldots,\pindex(\nodepath_k))$ must be a sub-chain of $\pichain$, i.e.
\[
(\pindex(\nodepath_1),\ldots,\pindex(\nodepath_k)) = \left( \pichain_m, \pichain_{m+1}, \ldots, \pichain_n \right)
\]
for some $1\leq m\leq n\leq |\pichain|$ with $n-m+1 = k$. But then $\nodepath$ satisfies the separating condition \eqref{eq:separting_tree_constraint} again since $\pichain$ is a separating chain by assumption.

Finally, the nodes introduced into $T$ by the operation $\underset{\pichain}{\mathrm{Graft}}(T)$ are all the unique children of their parents except for the node $\nodepath_1^{\mathrm{new}}$, which has parent node $\node^*$ in $T$. Recall that, by definition of the grafting operation, $\node^*$ is the terminal node of the longest rooted path $\nodepath^*$ such that $\pindex(\nodepath_i)=\pichain_i$ for all $1\leq i\leq |\nodepath^*|$.
Now suppose that there exists another child $\node'$ of $\node^*$ for which $\pindex(\nodepath_1^{\mathrm{new}})=\pindex(\node')$. Then it is easy to see that $\node'$ must be in $T$. In this case, $\nodepath'= \nodepath^*\circ \node'$ would be a rooted path in $T$ for which $\pindex(\nodepath'_i)=\pichain_i$ for all $1\leq i\leq |\nodepath'|$, where $\circ$ denotes the composition of paths. But $\nodepath'$ is longer than $\nodepath*$ by construction, and $\node^*$ is not the terminal node of $\nodepath'$. This contradicts the definition of $\node^*$, and thus we can conclude that $\node^*$ cannot have any other children $\node'$ with $\pindex(\node')=\pindex(\nodepath^{\mathrm{new}}_1)$. Thus, $\underset{\pichain}{\mathrm{Graft}}(T)$ satisfies constraint 3, and is a valid \DTree for $\qoi$, as claimed.
\end{proof}

With these lemmas, we can quickly prove that Algorithm \ref{alg:sensitivity_tree_from_chains_simple} produces valid \DTrees

\begin{proof}
  By Lemma \ref{lem:alg1_okay}, the $\pichain_i$ in Algorithm \ref{alg:sensitivity_tree_from_chains_simple} are all separating chains, and Lemma \ref{lem:tree_promotion_valid} therefore tells us that $\mathrm{Tree}(\pichain_1)$ is a separating tree. Finally, Lemma \ref{lem:grafting_valid} implies that, for $j=2,\ldots l$, $\underset{\pichain_j}{\mathrm{Graft}}(T)$ is a separating tree as well, completing the proof.
\end{proof}

\subsection{Proof of Lemma \ref{lem:st_preserved}}

\begin{proof}
  \textbf{Branching preserves separating \DTree constraints.}
  Because all nodes in $T$ are nonempty, branching cannot modify the root node. Thus, if $T$ is rooted at $\qoi$, $\branch{\node}{\pindex'}{T}$ will also be rooted at $\qoi$. To check the separating paths condition \eqref{eq:separting_tree_constraint}, we note that $\branch{\node}{\pindex'}{T}$ adds only a single node $\node'$ to $T$, and thus we need only check this condition for paths terminating in $\node'$. Let $\nodepath=\left( \node_1,\ldots, \node_k \right)$ be a path with $\node_k=\node'$. If $j=k$, then
  \[
  \pindex(\nodepath_{\geq j})\setminus \pindex(\nodepath_j)=\emptyset,
  \]
  and condition \eqref{eq:separting_tree_constraint} is trivially satisfied. Now define the path $\nodepath' = \left( \node_1,\ldots, \node_{k-1} \right)$. If $1\leq j\leq k-1$, then $\pindex\left(\nodepath_{\geq j}\right)=
\pindex(\nodepath'_{\geq j})$ since, by definition, $\pindex(\node_k) \subseteq \pindex(\node_{k-1})$. Clearly, we also have $\pindex(\nodepath_{\leq j})=\pindex(\nodepath'_{\leq j})$, so $\nodepath$ satisfies \eqref{eq:separting_tree_constraint} for $1\leq j\leq k-1$ if and only if $\nodepath'$ does. But $\nodepath'$ is a path in $T$, so $\nodepath'$ must satisfy \eqref{eq:separting_tree_constraint}. Finally, constraint 3 is satisfied by the definition of $\branch{\node}{\pindex'}{T}$.

  \textbf{Subdividing preserves separating \DTree constraints.}
  As with branching, constraints 1 and 3 are trivially satisfied by the definition of subdivision. In particular, subdivision cannot modify the root node of a tree, and does not add a node if the parent $\node_p$ already has a child $\node'$ with $\pindex(\node')=\pindex_s$ (where $\pindex_s$ is the subdividing set $\pindex_s = \pindex_p'\cup\pindex(\node_c)$ for some $\pindex_p'\subseteq \pindex(\node_p)$). Thus, we just need to verify that the separating paths condition \eqref{eq:separting_tree_constraint} is preserved. 
  
  Let $\node_s$ be the child of $\node_p$ in  $\subdivide{\node_p}{\pindex_p'}{\node_c}{T}$ with $\pindex_s = \pindex_p'\cup\pindex(\node_c)$. (This child node is unique since we know the subdivided tree still satisfies constraint 3 in Definition \ref{def:sensitivity_tree}.) It is easy to verify that any maximal path in $\subdivide{\node_p}{\pindex_p'}{\node_c}{T}$ not containing $\node_s$ and $\node_c$ is also a path in $T$, so it suffices to verify \eqref{eq:separting_tree_constraint} for maximal paths in $\subdivide{\node_p}{\pindex_p'}{\node_c}{T}$ containing $\node_s$ and $\node_c$.
  For any such path $\nodepath=\left( \node_1,\ldots,\node_k \right)$ of length $k$, let $i^*$ be the index such that $\node_{i^*} = \node_s$. Furthermore, define $\nodepath'=\left( \node_1,\ldots, \node_{i^*-1}, \node_{i^*+1},\ldots \node_k \right)$. Then $\nodepath'$ is a path in $T$, and thus $\nodepath'$ satisfies \eqref{eq:separting_tree_constraint} for all $1\leq j\leq k$. Our analysis is now divided into two cases.

  First, suppose $j\neq i^*$. If $j<i^*$, $\nodepath_{j} = \nodepath'_j$, $\pindex(\nodepath_{\leq j}) = \pindex(\nodepath'_{\leq j})$, and $\pindex(\nodepath_{\geq j}) = \pindex(\nodepath'_{\geq j})\cup \pindex(\node_s)$. But if $j<i^*$, then $\pindex(\nodepath'_{\geq j})\supset \pindex(\node_p)\supset\pindex_p'$ and $\pindex(\nodepath'_{\geq j})\supset \pindex(\node_c)$, so $\pindex(\nodepath'_{\geq j})\supset\pindex(\node_s)$. We thus conclude that $\pindex(\nodepath_{\geq j}) = \pindex(\nodepath'_{\geq j})\cup \pindex(\node_s)= \pindex(\nodepath'_{\geq j})$. By analogous reasoning, we can show that if $j>i^*$, then $\nodepath_j=\nodepath'_{j-1}$, $\pindex(\nodepath_{\geq j}) = \pindex(\nodepath'_{\geq j-1})$, and $\pindex(\nodepath_{\leq j}) = \pindex(\nodepath'_{\leq j-1})\cup \pindex(\node_s) = \pindex(\nodepath'_{\leq j-1})$. Since condition \eqref{eq:separting_tree_constraint} holds for $\nodepath'$ for all $1\leq j\leq k-1$, it therefore also holds for $\nodepath$ for all $j\neq i^*$.

  Next, consider $j=i^*$. In this case, we find that $\pindex(\nodepath_{\leq j}) = \pindex(\nodepath'_{\leq j})$, $\pindex(\nodepath_{\geq j}) = \pindex(\nodepath'_{\geq j})\cup \pindex_p'$, and $\pindex(\nodepath_j) = \pindex(\nodepath'_j)\cup \pindex_p'$. Thus, we have $\pindex(\nodepath_{\leq j})\setminus \pindex(\nodepath_j)\subseteq \pindex(\nodepath'_{\leq j})\setminus \pindex(\nodepath'_j)$ and $\pindex(\nodepath_{\geq j})\setminus \pindex(\nodepath_j)\subseteq \pindex(\nodepath'_{\geq j})\setminus \pindex(\nodepath'_j)$. It follows from this that
  \begin{equation}
    C_{M-\pindex(\nodepath_j)}\left( \pindex(\nodepath_{\leq j})\setminus \pindex(\nodepath_j) \right) \subseteq C_{M-\pindex(\nodepath'_j)}\left( \pindex(\nodepath'_{\leq j})\setminus \pindex(\nodepath'_j) \right).
  \end{equation}
  Putting this together, we find that 
  \begin{equation}
    C_{M-\pindex(\nodepath_j)}\left( \pindex(\nodepath_{\leq j})\setminus \pindex(\nodepath_j) \right)\cap \left( \pindex(\nodepath_{\geq j})\setminus \pindex(\nodepath_j) \right)\subseteq C_{M-\pindex(\nodepath'_j)}\left( \pindex(\nodepath'_{\leq j})\setminus \pindex(\nodepath'_j) \right)\cap \left( \pindex(\nodepath'_{\geq j})\setminus \pindex(\nodepath'_j) \right) = \emptyset,
  \end{equation}
  so we conclude that $\nodepath$ satisfies \eqref{eq:separting_tree_constraint} for all $1\leq j\leq k$, as required.

  \textbf{Merging preserves separating \DTree constraints.}
  Because merging is defined as a sequence of grafting and branching operations, it follows from the above and Lemma \ref{lem:grafting_valid} that merging preserves all \DTree constraints.
\end{proof}

\subsection{Uncertainty Index Estimation}

\subsubsection{Random Forest Regressor}
We use a random forest regressor as our conditional expectation estimator $\Ehat{\pindex}{\mathcal{S}_E}{\npset}$ throughout all examples. Specifically, we use the implementation of random forests in the \texttt{ranger} package \citep{ranger}. For simplicity, and to demonstrate that the usefulness of our estimator $\hat{U}^2$ does not depend on careful tuning of the underlying conditional expectation estimator, we use all of the default settings in \texttt{ranger}. These defaults are as follows.
\begin{enumerate}
  \item Number of trees: 500
  \item Fraction of variables considered at each split: 1/3
  \item Bootstrap method: full data set sampled with replacement
  \item Splitting criterion: node impurity (measured by variance)
\end{enumerate}

We refer the reader to \cite{ranger} for full details on its implementation of random forests.  

\subsubsection{Asymptotic Unbiasedness of Conditional Variance Estimator}
In this section, all expectations, variances, and distributions are implicitly conditional on $\data$. We suppress the explicit dependence on $\data$ for notational clarity. Recall that for $\pindex\subseteq\pindext$ and $\npset'\in \RR^{\pdim}$, we use $\Ehat{\pindex}{\mathcal{S}}{\npset'}$ to denote an estimator of the conditional expectation $\E\left[ \qoi\mid \npi=\npi' \right]$ based on a sample $\mathcal{S}=\{\nptotal^{(s)}\}_{s=1}^{ |\mathcal{S}| }$, where $\qoi = \nptotal_{i_{\qoi}}$. Given such an estimator, we then construct the data-splitting expected conditional variance estimator $\CVhat{\pindex}{\hat{E}}{\mathcal{S}}$ as follows. First, we partition the sample as $\mathcal{S}=\mathcal{S}_E\cup\mathcal{S}_V$, where $|\mathcal{S}_E|=|\mathcal{S}_V|$ and $\mathcal{S}_E\cap\mathcal{S}_V=\emptyset$ (assuming $|\mathcal{S}|$ even). Then our estimator is defined as
\begin{equation}
  \label{eq:var_est_repeat}
  \CVhat{\pindex}{\hat{E}}{\mathcal{S}} = \frac{2}{|\mathcal{S}|}\sum_{\npset^{(s)}\in \mathcal{S}_V} \left( \qoi^{(s)} - \Ehat{\pindex}{\mathcal{S}_E}{\npset^{(s)}} \right)^2.
\end{equation}
 
We now establish a sufficient condition for the estimator \eqref{eq:var_est_repeat} to be asymptotically unbiased.

\begin{lemma}\label{lem:var_est_unbiased}
Let $\mathcal{S}$ be an i.i.d. sample from $p(\nptotal)$. Suppose that, for $\pindex\subseteq\pindext$, the average mean squared error (MSE) of our conditional expectation estimator $\Ehat{\pindex}{\mathcal{S}_E}{\npset}$ vanishes in the sense that
\[
  \E\left[ \left( \Ehat{\pindex}{\mathcal{S}_E}{\npset} - \E\left[ \qoi\mid\npi \right] \right)^2 \right]=\E\left[ \E\left[ \left( \Ehat{\pindex}{\mathcal{S}_E}{\npset} - \E\left[ \qoi\mid\npi \right]\right)^2\;\Big\vert\; \npi \right] \right]\to 0 \text{ as } |\mathcal{S}|\to\infty,
\] 
where the inner expectation above is the point-wise MSE of $\Ehat{\pindex}{\mathcal{S}_E}{\npset}$ and the outer expectation averages over inputs $\npset$, with $\npset\ind \mathcal{S}$.
Then the data-splitting expected variance estimator \eqref{eq:var_est_repeat} is asymptotically unbiased for the true expected variance $\E\left[ \mathrm{Var}\left( \qoi\mid \npi \right) \right]$.
\end{lemma}

\begin{proof}
  We take $\pindex\subseteq\pindext$ as fixed. We start by expanding each term in the conditional variance estimator:
  \begin{align*}
    \left( \qoi^{(s)} - \Ehat{\pindex}{\mathcal{S}_E}{\npset^{(s)}} \right)^2 = &\underbrace{\left( \qoi^{(s)} - \E\left[ \qoi\mid \npi=\npi^{(s)} \right] \right)^2}_{A} + \underbrace{\left( \Ehat{\pindex}{\mathcal{S}_E}{\npset^{(s)}} - \E\left[ \qoi\mid \npi= \npi^{(s)} \right] \right)^2}_{B}\\
    &\hspace{0.3in}- 2\underbrace{\left( \qoi^{(s)} - \E\left[ \qoi\mid \npi=\npi^{(s)} \right] \right)\left( \Ehat{\pindex}{\mathcal{S}_E}{\npset^{(s)}} - \E\left[ \qoi\mid \npi= \npi^{(s)} \right] \right)}_{C}.
  \end{align*}
  The expected value of term $A$ is exactly the estimand. Specifically:
  \begin{align*}
    \E\left[ A \right] &= \E\left[ \E\left[ A\mid \npi=\npi^{(s)} \right]\right]\\
    &= \E\left[ \E\left[ \left( \qoi^{(s)} - \E\left[ \qoi\mid \npi=\npi^{(s)} \right] \right)^2 \;\Big\vert\; \npi = \npi^{(s)} \right] \right]\\
    &= \E\left[ \mathrm{Var}\left( \qoi\mid\npi \right) \right].
  \end{align*}

  Next, we observe that term $C$ vanishes in expectation. Specifically, since $\mathcal{S}_E$ is independent of $\npset^{(s)}\in \mathcal{S}_V$, it follows that
  \[
    \left( \qoi^{(s)} - \E\left[ \qoi\mid \npi=\npi^{(s)} \right] \right) \ind \left( \Ehat{\pindex}{\mathcal{S}_E}{\npset^{(s)}} - \E\left[ \qoi\mid \npi= \npi^{(s)} \right] \right)\;\Big\vert\; \npi^{(s)}.
  \]
  Therefore, iterating expectations and applying the above conditional independence, we can write the expected value of $C$ as
  \begin{align*}
    &\E\left[ \left( \qoi^{(s)} - \E\left[ \qoi\mid \npi=\npi^{(s)} \right] \right)\left( \Ehat{\pindex}{\mathcal{S}_E}{\npset^{(s)}} - \E\left[ \qoi\mid \npi= \npi^{(s)} \right] \right) \right]\\
    &= \E\left[ \underbrace{\E\left[ \left( \qoi^{(s)} - \E\left[ \qoi\mid \npi=\npi^{(s)} \right] \right) \;\Big\vert\;\npi=\npi^{(s)} \right]}_{C'}\E\left[ \left( \Ehat{\pindex}{\mathcal{S}_E}{\npset^{(s)}} - \E\left[ \qoi\mid \npi= \npi^{(s)} \right] \right)\;\Big\vert\;\npi=\npi^{(s)} \right] \right].
  \end{align*}
  Now the factor $C'$ is just
  \[
  C' = \E\left[\qoi\mid \npi=\npi^{(s)}\right] - \E\left[\qoi\mid \npi=\npi^{(s)}\right] = 0,
  \]
  so in fact $\E\left[ C \right] = 0$.

  Finally, $\E\left[ B \right]$ is exactly the average mean squared error of the estimator $\Ehat{\pindex}{\mathcal{S}_E}{\npset^{(s)}}$ over inputs $\npset^{(s)}$. If this average MSE vanishes as $|\mathcal{S}|\to\infty$, then the above shows that 
  \begin{align*}
    \E\left[ \CVhat{\pindex}{\hat{E}}{\mathcal{S}} \right] &= \frac{2}{|\mathcal{S}|}\sum_{\npset^{(s)}\in \mathcal{S}_V} \E\left[ \left( \qoi^{(s)} - \Ehat{\pindex}{\mathcal{S}_E}{\npset^{(s)}}\right)^2 \right]\\
    &= \E\left[ \mathrm{Var}\left( \qoi\mid\npi \right) \right] + \E\left[ B \right]\\
    &\to \E\left[ \mathrm{Var}\left( \qoi\mid\npi \right) \right],
  \end{align*}
  as claimed.
\end{proof}

In practice, most practical estimators $\hat{E}$ will \textit{not} have asymptotically vanishing average MSE, and Lemma \ref{lem:var_est_unbiased} will not apply exactly. Instead, any residual bias in the estimator $\hat{E}$ will bias our estimate of the expected variance $\E\left[ \mathrm{Var}(\qoi\mid\npi) \right]$, and of the uncertainty index, upwards. In this sense, the estimator \eqref{eq:var_est_repeat} is conservative: It may yield ``false negatives'' by underestimating the uncertainty reduction associated with a set $\npi$, but it is unlikely to yield ``false positives''.

\subsection{Causal Latent Factor Model Details}\label{sec:lfm_supp}

This section provides all details of the model specification for the latent factor model used in the example of Section 3, as well as a description of how the posterior predictive data were generated for a new unit and at new time points.

\subsubsection{Model Specification}

The latent factor model in Section \ref{sec:ex1_sc} is a model of panel data $\datam \in \RR^{N\times T}$. Recall that the factor model decomposes this data as
\begin{equation}
  \label{eq:LFM_supp}
  \datam = \lfload\lffac + \effect + \epsilonv,
\end{equation}
where $\lffac$ and $\lfload$ are the latent factors and corresponding loadings, $\effect$ is the matrix of treatment effects (zero for all untreated units), and $\epsilonv$ the matrix of independent, idiosyncratic normal errors. We now give detailed structural and distributional assumptions for these parameters.

Most importantly, the factor loadings $\lfload\in \RR^{N\times K}$ are assumed to be lower triangular with nonnegative diagonal. This assumption is made to avoid nonidentification of the parameters $\lfload$ and $\lffac$, which otherwise occurs due to the fact that $\datam$ depends on these parameters only through their product $\lfload\lffac$ \citep{bayes_lfm_id}. Critically, this identifying assumption does \textit{not} constrain the product $\lfload\lffac$. 

We now state our parametric and distributional assumptions. 
For each unit $1\leq i\leq N$, the factor loadings $\lfloadc_{i\cdot}$ control both the scale of $y_{i\cdot}$ (through the norm $\|\lfloadc_{i\cdot}\|$), and the relative importance of each latent factor to $y_{i\cdot}$ (through the relative square magnitudes of the components $\lfloadc_{ij}^2/\|\lfloadc_{i\cdot}\|^2$). We parametrize $\lfload$ to allow specifying priors for these two components separately. Specifically, we decompose $\lfload$ as
\begin{equation*}
  \lfload = \lfscale\lfcorfac,
\end{equation*}
where: 
\begin{itemize}
  \item $\lfscale\in\RR^{N\times N}$ is a diagonal matrix with $\lfscalec_{ii}\stackrel{\mathrm{def}}=\sigma'_i>0$ controlling the scale of unit $y_{i\cdot}$.
  \item $\lfcorfac\in\RR^{N\times K}$ is the Cholesky factor of a correlation matrix, i.e. a lower diagonal matrix with positive diagonal elements (i.e. $\lfcorfacc_{kk}>0$ for $1\leq k\leq K$) and unit vector rows (i.e. $\|\lfcorfacc_{i\cdot}\|=1$ for $1\leq i\leq N$). Note that the squared components of each row sum to $1$: $\sum_{k=1}^K\left[ \lfcorfacc_{ik} \right]^2$. Thus, these squared components express the (standardized) unit $y_{i\cdot}$ as a weighted average of the latent factors.
\end{itemize}

For the scale parameters $\sigma'_i$, we assign priors
\begin{equation*}
  \sigma'_i \stackrel{iid}{\sim}\mathrm{InverseGamma}\left( 11, 10s \right),
\end{equation*}
where $s$ is a fixed scale hyperparameter which we set to the mean standard deviation of the units $y_i$. This prior has mean $s$ and standard deviation $s/3$.

Subject to the lower triangular and positive diagonal constraints, we take each row of the matrix $\lfcorfac$ to have a uniform prior. Since the rows must be unit vectors, each row is constrained to a compact space, and thus a uniform distribution is a proper prior.

The $K$ latent factors are then assigned independent AR(1) process priors with autocorrelations $\rho_k$ and stationary variances fixed to be $1$:
\begin{gather*}
  \lffacc_{k1} \sim \mathrm{normal}\left( 0, \sqrt{1 - \rho_k^2} \right) \text{ for } 1\leq k\leq K,\\
  \lffacc_{kt}\sim \mathrm{normal}\left( \rho_k\lffacc_{k(t-1)}, \sqrt{1 - \rho_k^2} \right) \text{ for } 1\leq k\leq K \text{ and } 2\leq t\leq T\\
  \rho_k \stackrel{iid}{\sim} \mathrm{uniform}(0, 1).
\end{gather*}

Finally, the nonzero treatment effects $\effects_t$ and idiosyncratic errors $\epsilonv$ receive centered normal priors:
\begin{gather*}
  \effects_t\sim\mathrm{normal}\left( 0, \sigma_{\effects} \right)\text{ for } 1\leq t\leq T-t_{\delta}\\
  \epsilonc_{it} \sim \mathrm{normal}\left( 0, f\times \sigma'_i \right)\text{ for } 1\leq i\leq N\text{ and } 1\leq t\leq T,
\end{gather*}
where $\sigma'_i$ is the $i^{\mathrm{th}}$ diagonal element of the scale matrix $S$ in our decomposition of $\lfload$, and $f$ is a fixed scaling hyperparameter. This parameterization implies that $\sigma'_i$ control the \textit{overall} scale of $\data_{i\cdot}(0)$ (combining the latent and idiosyncratic components), and $f$ controls the relative size of the idiosyncratic component. We set $f$ to $0.025$ to express the prior belief that the idiosyncratic component should typically be no more than $5\%$ of the scale of the latent component. Choosing a small $f$ is beneficial for identification, since either the latent or idiosyncratic component can account for a majority of the variation in the data absent some constraint.

\subsubsection{Posterior Predictive Sampling}

In order to generate hypothetical data at earlier time points and for an additional unit, we proceed as follows. First, we sample each of the latent factors at $T^*$ earlier time points. Because the AR(1) prior for the latent factors $\lffac$ is stationary and Gaussian, it is time-reversible. Therefore, we can sample the latent factors at earlier times by simply iteratively applying the autoregressive equations in reverse. Specifically, let $\lffacnew\in\RR^{K\times T^*}$ represent the latent factors at the earlier times in reverse temporal order (i.e. with $\lffacnewc_{\cdot 1}$ immediately preceding $\lffacc_{\cdot 1}$ in time). Then, for each $1\leq k\leq K$, we simulate $\lffacnewc_{k\cdot}$ as:
\begin{gather}
  \label{eq:lffac_earlier_sample}
  \lffacnewc_{k1} = \rho_k\lffacc_{k1} + \xi_{k1},\\
  \lffacnewc_{kt} = \rho_k\lffacnewc_{k(t-1)} + \xi_{kt}\text{ for } 2\leq t\leq T^*,\\
  \xi_{kt} \stackrel{iid}{\sim} \mathrm{normal}\left( 0, \frac{1}{\sqrt{1-\rho_k^2}} \right),\\
  (\lffacc_{k1}, \rho_{k}) \sim p(\lffacc_{k1}, \rho_k\mid \datam).
\end{gather}
Hypothetical data for the $N$ observed units are then generated at these earlier times according to the latent factor model \eqref{eq:LFM_supp}. For each unit $1\leq i\leq N$ and time $1\leq t\leq T^*$, the hypothetical data (also in reverse chronological order) is given by:
\begin{gather}
  \label{eq:earlier_units}
  \datapot_{i t} = \lfloadc_{i \cdot}\lffacnewc_{\cdot t} + \epsilon_{i t},\\
  \epsilon_{i t}\sim\mathrm{normal}\left( 0, f\times \sigma_i \right)\\
  (\lfloadc_{i\cdot}, \sigma_i) \sim p(\lfloadc_{i\cdot}, \sigma_i\mid \datam).
\end{gather}
Note that the treatment effect term is omitted since all times at which we generate hypothetical data are pre-treatment times.

Finally, to generate posterior predictive samples for the new unit, we first sample new loadings from the prior distribution:
\begin{equation*}
  \begin{split}
    \lfloadnewc_{\mathrm{new}} &= \overline{\sigma} \times u,\\
    \overline{\sigma} &= \frac{1}{N}\sum_{i=1}^N\sigma_i\\
    u &\sim \mathrm{uniform}\left( S^{k-1} \right),
  \end{split}
\end{equation*}
i.e. we multiply the average unit scale by a uniformly distributed unit vector in $\RR^K$;
Given these, the new unit is simulated at each time point $1 \leq t\leq T+T^*$ according to
\begin{equation*}
  \begin{split}
    \datapot_t &= \lfloadnewc_{\mathrm{new}}^T\lffacnewc_{\cdot t} + \epsilonc'_t\;\text{ for } 1\leq t\leq T^*,\\
    \datapot_{T^*+t} &= \lfloadnewc_{\mathrm{new}}^T\lffacc_{\cdot t} + \epsilonc'_t\; \text{ for } 1\leq t\leq T,\\
  \epsilonc'_t &\sim \mathrm{normal}\left( 0, f\times \overline{\sigma} \right),\\
  \left( \lffac,\lffacnew,\overline{\sigma} \right) &\sim p\left(  \lffac,\lffacnew,\overline{\sigma}\mid \datam \right).
  \end{split}
\end{equation*}

\subsection{Polling Model Details}\label{sec:polling_supp}

For full details of the polling model analyzed in section 4, including discussion of modeling choices and parametrization, see \cite{gelman_polling}.

\subsection{Almost-Sparse Markov Graphs}\label{sec:almost_sparse_supp}

The examples presented in this paper all involve models with a Markov graph $M_{\mathcal{F}}$ that is sufficiently sparse to have useful local structure, which can then be visualized using our proposed trees. Many commonly encountered models lack this pleasant property, which limits the usefulness of our visualization technique as defined in the main text. In this section, we describe this problem in greater detail and introduce a simple extension of our method which can be fruitfully applied to ``almost'' sparse graphs (to be adequately defined below).

For our purposes, a Markov graph is sparse if every vertex (corresponding to a single, scalar unknown quantity $\npscalar$) is connected to only a few other vertices. In a non-sparse graph, there will be at least one vertex that is connected to a large fraction of the other vertices -- potentially all other vertices. The existence of such highly-connected vertices tends to shorten the shortest paths between any two vertices.

It follows from our minimal separator construction \ref{eq:minimal_separators} that the length of a branch in a \DTree generated by Algorithm 2 is upper bounded by the length of any shortest path connecting the root vertex to a vertex in the leaf node. Consequently, when applied to non-sparse Markov graphs, Algorithm 2 will tend to generate shallow trees with fewer nodes, limiting our ability to discern relationships among the quantity of interest and other unknown quantities.

In some cases, this problem is unavoidable. However, in many cases, a Markov graph may be ``almost'' sparse in the sense that a majority of the vertices have few neighbors, while a few vertices have many neighbors. We refer to the former as local vertices and the latter as global vertices. In this case, the global vertices may block our ability to discern useful local connections among the local vertices. A natural solution in cases like this is to condition away the global vertices.

Specifically, suppose that we have divided all unknowns into local parameters $\npset^L = \npset_{\pindex_L}$ and global parameters $\npset^G = \npset_{\pindex_G}$, with
\[
  \pindex_{L}\cup \pindex_{G} = \pindext,\;\pindex_L\cap\pindex_G=\emptyset,\;\text{and}\; |\pindex_G|\ll|\pindext|.
\]
Then our posterior distribution is $p(\npset^L, \npset^G)$ (suppressing dependence on the data for notational simplicity). If $p(\npset^L,\npset^G)$ factorizes over $\mathcal{F}$, then it is easy to see that $p(\npset^L\mid \npset^G)$ factorizes over
\[
  \mathcal{F}' = \left\lbrace f\cap \pindex_L \;\Big\vert\; f\in \mathcal{F} \right\rbrace,
\]
and that the corresponding Markov graph is $M_{\mathcal{F'}} = M_{\mathcal{F}}-\pindex_{G}$ (i.e. the graph obtained from $M$ by deleting all vertices in $\pindex_G$). Consequently, if $\pindex_G$ corresponds to the global vertices in $M_{\mathcal{F}}$, then $p(\npset^L\mid\npset^G)$ will have a sparse Markov graph $M_{\mathcal{F}}-\pindex_G$, to which we can then fruitfully apply Algorithm 2. 

This strategy leads to a simple extension of Algorithm 2. Let $\mathrm{PostTree}\left( M, \qoi, \{\npsetl_i\}_{i=1}^l, \gamma \right)$ denote the procedure defined in Algorithm 2, which outputs a \DTree $T$. Let $\node^{\mathrm{int}}(T)$ denote the internal nodes of $T$, i.e. all non-root and non-leaf nodes. Algorithm \ref{alg:tree_alg_global} gives the modified procedure. In words, we simply apply Algorithm 2 to $M-\pindex_G$, and then modify the resulting \DTree by inserting the global indices $\pindex_G$ at each internal node in $\node^{\mathrm{int}}(T)$. It is easy to see that the resulting tree $T$ is still a separating tree with respect to the original Markov graph $M$.

\begin{algorithm}
  \begin{algorithmic} \onehalfspacing
  \Require Undirected graph $M$, root quantity $\qoi$, leaf nodes $\{\nodel_i\}_{i=1}^l$, threshold $\gamma\in(0,1)$, global indices $\pindex_G$.
  \State $M' \gets M-\pindex_G;$
  \State $T \gets \mathrm{PostTree}\left( M', \qoi, \{\nodel_i\}_{i=1}^l, \gamma \right);$
  \For{$\node\in \node^{\mathrm{int}}(T)$}
    \State $\pindex(\node) \gets \pindex(\node) \cup \pindex_G$;
  \EndFor
  \State \Output $T$
  \end{algorithmic}
  \caption{\DDTree Construction with Global Parameters}
  \label{alg:tree_alg_global}
\end{algorithm} 

Algorithm \ref{alg:tree_alg_global} provides a theoretical solution to the problem of finding useful \DTrees in non-sparse models, but two practical problems remain. The first is the problem of how to specify the global parameters $\npset^G$. As the statement of Algorithm \ref{alg:tree_alg_global} indicates, this is entirely user specifiable. While the concept of a global parameter was not given a precise definition, in most cases, we expect there to be a fairly clear delineation between the parameters with ``few'' neighbors and the parameters with ``many''. In such cases, simple heuristics can identify the candidate global parameters automatically from the graph $M$.

The second problem arises when $\pindex_G$ is large. For any internal node $\node$ in $T$, we will have $\pindex_G\subset \pindex(\node)$, so the size of $\pindex_G$ lower bounds the size of the node index sets $\pindex(\node)$. In this case, the size of $\pindex_G$ can make the difference between $\hat{U}_{\pindex(\node)}$ being tractable and intractable to compute. By constructing minimal separating sets at each internal node, Algorithm \ref{alg:sensitivity_tree_from_chains_simple} was designed to suppress the size of the node index sets $\pindex(\node)$. But by injecting $\pindex_G$ in each internal node, Algorithm \ref{alg:tree_alg_global} lacks this virtue, and may thus reasonably be seen as only a partial solution to the problem.

When $\pindex_G$ is large, another possible solution is to re-run the original posterior inference algorithm (e.g. Markov chain Monte Carlo, as used in our examples) conditioning on a particular value of $\npset^G$. In this case, we obtain samples directly from $p(\npset^L\mid\npset^G)$, and can thus drop $\pindex_G$ from the computation of the $\hat{U}_{\pindex(\node)}$. The result, however, is only an approximation to the output of Algorithm \ref{alg:tree_alg_global} since, in this case, we only condition on a single value of $\npset^G$ (rather than taking an expectation over the posterior).

\end{document}